%% file: negotiationsIItechrep-retraction-arxiv.tex
\documentclass{llncs}

\usepackage{amssymb}
\usepackage{eucal}
\usepackage{graphicx}
\usepackage{amsmath}
\usepackage{color}
\usepackage{stmaryrd}
\usepackage{rotating}
\usepackage{xspace}
\usepackage{eepic}
\usepackage{epsfig}
\usepackage{latexsym} 
\usepackage{exscale}
\usepackage{verbatim}

\newcommand{\agents}{A}

\newcommand{\outc}{R}
\newcommand{\outrel}{\delta}

\newcommand{\vx}{\vec{x}}
\newcommand{\next}{{\cal X}}

\newcommand{\trf}[1]{\langle #1 \rangle}
\newcommand{\N}{{\cal N}}
\newcommand{\F}{{\cal F}}
\renewcommand{\L}{{\cal L}}
\newcommand{\Ta}{{\it Ta}}
\newcommand{\Out}{{\it Out}}

\newcommand{\jd}[1]{\textcolor{black}{#1}}
\newcommand{\jde}[1]{\textcolor{black}{#1}}
\newcommand{\je}[1]{\textcolor{black}{#1}}

\def\by#1{\mathop{{\hbox{\setbox0=\hbox{$\scriptstyle{#1\quad}$}{$\buildrel{\quad\scriptstyle{#1}\quad}\over{\hbox to \wd0{\rightarrowfill}}$}}}}}

\def\bytwo#1#2{%
\setbox0=\hbox{$\scriptstyle{#1\quad}$}%
\mathrel{\mathop{\hbox to \wd0{\rightarrowfill}}\limits^{#1}_{#2}}%
}

\newcounter{zeile}
\newbox\kasten
\setbox\kasten=\hbox{00}
\let\graph=\par
{\catcode`\[=\active\catcode`\>=\active
\gdef	\addto#1#2{$#1\leftarrow#1\cup\{#2\}$}

\gdef	\while#1{[while] $#1$ [do]}

\gdef\algo{\catcode`\~=\active
	\catcode`\[=\active\catcode`\>=\active \setcounter{zeile}{0}
	\def\par{ \refstepcounter{zeile}
                 \graph\noindent\kern\wd\kasten%
		  \llap{{\small\thezeile}}
                \quad}
	\def[##1]{{\bf##1}}\def~##1~{\mathchar"405B##1\mathchar"505D}
	\def>{\quad}\obeylines}}

\begin{document}

\pagestyle{empty}  

\author{Javier Esparza\inst{1} \and J\"org Desel\inst{2}} 

\title{On Negotiation as Concurrency Primitive II: \\Deterministic Cyclic Negotiations}
\institute{Fakult\"at f\"{u}r Informatik, Technische Universit\"{a}t M\"{u}nchen, Germany \and
Fakult\"at f\"ur Mathematik und Informatik, FernUniversit\"at in Hagen, Germany}
\maketitle

\setlength{\belowcaptionskip}{-2pt}

\pagestyle{headings}  

\begin{abstract}
We continue our study of negotations, a concurrency model with multiparty negotiation as primitive. In a previous paper \cite{negI} we have provided a correct and complete set of reduction rules for sound, acyclic, and (weakly) deterministic negotiations. In this paper we extend this result to all deterministic negotiations, including cyclic ones. We also show that this set of rules allows one to decide soundness and to summarize negotiations in polynomial time.
\end{abstract}

\section*{Retraction}

{\em Unfortunately, while preparing a journal version of this contribution we have discovered that Lemma \ref{lem:soundsubneg} is wrong. Since the lemma is used in line 3 of the reduction procedure described in Section \ref{subsec:redproc}, the procedure is also incorrect.

We have uploaded a preprint of the journal version to arXiv
(\url{http://arxiv.org/abs/1612.07912}). The preprint 
contains a counterexample to Lemma \ref{lem:soundsubneg} (see Section 7.7). It also presents a corrected reduction procedure with the same characteristics as the incorrect one. So, fortunately, while the specific reduction procedure presented in this contribution is wrong, our main results still hold:
\begin{itemize}
\item The merge, iteration, and shortcut reduction rules are complete for the class of sound and deterministic negotiations, i.e., they completely reduce all and only the sound deterministic negotiations to an atomic negotiation.
\item There exists a polynomial $p(x)$ such that every sound deterministic negotiation of size $n$ can be reduced to an atomic negotiation by means of at most $p(n)$ applications of the reduction rules.
\item There exists a polynomial-time algorithm that, given a sound deterministic negotiation of size $n$, constructs a reduction sequence of length at most $p(n)$.
\end{itemize}
}

\section{Introduction}

Negotiation has long been identified as a paradigm for process interaction 
\cite{davis1983negotiation}. It has been applied to different problems
(see e.g. \cite{winsborough2000automated,atdelzater2000qos}), and
studied on its own \cite{jennings2001automated}. However, it
has not yet been studied from a concurrency-theoretic point 
of view. In \cite{negI} we have initiated a study 
of negotiation as communication primitive. 

Observationally, a negotiation is an interaction in which several partners come 
together to agree on one out of a number of possible outcomes
(a synchronized nondeterministic choice). 
In \cite{negI} we have introduced {\em negotiations}, a Petri-net like concurrency model 
combining multiparty ``atomic'' negotiations 
or {\em atoms} into more complex {\em distributed negotiations}. 
Each possible outcome of an atom has associated a state-transformer.
Negotiation partners enter the atom in certain initial states, and leave
it in the states obtained by applying to the initial states
the state-transformer of the outcome agreed upon. Atoms are combined into 
more complex, distributed negotiations,
by means of a next-atoms function that determines for each atom, 
negotiating agent, and outcome, the set of atoms the agent is ready 
to engage in next if the atom ends with that outcome.

Negotiations are close to a colored version of van der {\jd Aalst's} 
{\em workflow nets} \cite{aalst}. Like in workflow nets, distributed negotiations can 
be {\em unsound} because of deadlocks or livelocks. The {\em soundness} 
problem consists of deciding if
a given negotiation is sound. Moreover, a sound negotiation
is equivalent to a single atom whose {\jd state transformation function} 
determines the possible final internal states of all parties
as a function of their initial internal states.  The {\em summarization problem} 
consists of computing such an atomic negotiation, called a {\em summary}. 

Negotiations can simulate 1-safe Petri nets (see 
the arXiv version of \cite{negI}), which proves that the soundness problem 
and (a decision version of) the summarization problem are, unsurprisingly, 
PSPACE-complete. For this reason we have studied in \cite{negI}
two natural classes: {\em deterministic} and {\em weakly deterministic} negotiations. 
Only deterministic negotiations are relevant for this paper.
Loosely speaking, a negotiation is deterministic if, for each agent and each outcome of
an atomic negotiation, the next-atom function yields only one next atom, i.e., each 
agent can always engage in one atom only. 

In particular, we have shown in \cite{negI} that the soundness and summarization problems 
for {\em acyclic} deterministic negotiations can be solved in polynomial time.
(Notice that the state space of a deterministic negotiation can be exponentially larger
then the negotiation itself). 
The algorithm takes the graphical representation of a reduction procedure in which 
the original negotiation is progressively reduced to a simpler one by means 
of a set of reduction rules. Each rule preserves soundness and summaries (i.e., the negotiation 
before the application of the rule is sound if{}f the negotiation after the 
application is sound, and both have the same summary). Reduction rules have been extensively applied to 
Petri nets or workflow nets, but most of this work has been devoted to the 
liveness or soundness problems \cite{DBLP:conf/ac/Berthelot86,DBLP:conf/apn/Haddad88,DBLP:journals/ppl/HaddadP06,DBLP:journals/tcs/GenrichT84,Desel:1995:FCP:207572}, and many rules, 
like for example the linear dependency rule 
of \cite{Desel:1995:FCP:207572}, do not preserve summaries.

In \cite{negI} we conjectured that the addition of a simple rule allowing one
to reduce trivial cycles yields a complete set of rules for all sound 
deterministic negotiations. 
In this paper we prove this result, and show that the number of 
rule applications required to summarize a negotiation is still polynomial.

While the new rule is very simple, the proof of our result is very involved. 
It is structured in 
several sections, and some technical proofs have been moved to an appendix.
More precisely, the paper is structured as follows. Sections \ref{sec:synsem}
and \ref{sec:determ} presents the main definitions of \cite{negI} in compact form. 
Section \ref{sec:redrules} introduces our set of three reduction rules.
Section \ref{sec:compl} proves that the rules summarize all sound
deterministic negotiations. Section \ref{sec:complexity} proves that the 
summarization requires a polynomial number of steps.

\section{Negotiations: Syntax and Semantics}
\label{sec:synsem}

We recall the main definitions of \cite{negI}, and refer to this paper for more 
details.

We fix a finite set $\agents$ of {\em agents}. 
Each agent $a \in \agents$ has a (possibly infinite) nonempty set $Q_a$ of {\em internal states}. 
We denote by $Q_\agents$ the cartesian product $\prod_{a \in \agents} Q_a$. 
{\jd A {\em transformer} is a left-total relation $\tau \subseteq Q_\agents \times Q_\agents$}. 
{\jd Given $S \subseteq \agents$,
we say that a transformer $\tau$ is an {\em $S$-transformer} if, for each $a_i \notin S$,  
$\left((q_{a_1} , \ldots , q_{a_i}, \ldots , q_{a_{|A|}}), (q'_{a_1} , \ldots ,q'_{a_i}, \ldots ,  q'_{a_{|A|}})\right)\in\tau$ implies $q_{a_i} = q'_{a_i}$.}
So an $S$-transformer  only transforms the internal
states of  agents in $S$. 

\begin{definition}
A {\em negotiation atom}, or just an {\em atom}, is a triple $n=(P_n, \outc_n,\outrel_n)$,
 where $P_n \subseteq \agents$ is a nonempty set of {\em parties}, 
$\outc_n$ is a finite, nonempty set of {\em outcomes}, {\jd and 
 $\outrel_n$  is a mapping assigning to each outcome $r$} in $\outc_n$ 
a $P_n$-transformer $\outrel_n (r)$.
\end{definition}

\noindent Intuitively, if the  
states of the agents before a negotiation $n$ are given by a tuple $q$ 
and the outcome of the negotiation is $r$, then the agents  change
their  states to $q'$ for some $(q,q') \in \delta_n (r)$. 

For a simple example, consider a negotiation atom $n_\texttt{FD}$ with 
parties \texttt{F} (Father) and \texttt{D} (teenage Daughter). The goal of the
negotiation is to determine whether \texttt{D} can go to a party, and the time 
at which she must return home. The possible outcomes are
$\texttt{yes}$ ($\texttt{y}$) and $\texttt{no}$. 
Both sets $Q_{\texttt{F}}$ and $Q_{\texttt{D}}$ contain a state 
$\bot$ plus a state $t$ for every time $T_1 \leq t \leq T_2$ in a 
given interval $[T_1,T_2]$. {\jd Initially, \texttt{F} is in state $t_f$ and \texttt{D} in state $t_d$.} 
The  transformer $\delta_{n_\texttt{FD}}$ is given by
$$
\begin{array}{rcl}
\delta_{n_{fd}} (\texttt{yes})  & = & \left\{ \left( (t_f, t_d)  , (t,t)\right) \; \mid \;  t_f \leq t \leq t_d \vee t_d \leq t \leq t_f \right\} \\
\delta_{n_{fd}} (\texttt{no}) & = &  \left\{\left((t_f, t_d) ,  (\bot , \bot)\right) \; \right\}  
\end{array} 
$$


\subsection{Combining atomic negotiations}

A negotiation is a composition of atoms.  We add a {\em transition function} $\next$ that assigns to every 
triple $(n,a,r)$ consisting of an atom $n$, a participant $a$ of $n$, and an outcome $r$ of $n$ a set 
$\next(n,a,r)$ of atoms. Intuitively, this is the set of atomic negotiations 
agent $a$ is ready to engage in after the atom $n$,  if the outcome of $n$ is $r$. 

\begin{definition}
Given a {\jd finite} set of atoms $N$, let  $T(N)$ denote the set of triples $(n, a, r)$ such that 
$n \in N$, $a\in P_n$, and $r \in \outc_n$. 
A {\em negotiation} is a tuple ${\cal N}=(N, n_0, n_f, \next)$, where 
$n_0, n_f \in N$ are the {\em initial} and {\em final} atoms, and  $\next \colon T(N) \rightarrow 2^N$ is the {\em transition function}. Further, ${\cal N}$ satisfies the following properties: 

(1) every agent of $\agents$ participates in both $n_0$ and $n_f$; 

(2) for every $(n, a, r) \in T(N)$: $\next(n, a, r)= \emptyset$ if{}f $n=n_f$.

\end{definition}


\begin{figure}[t]
\centerline{\scalebox{0.35}{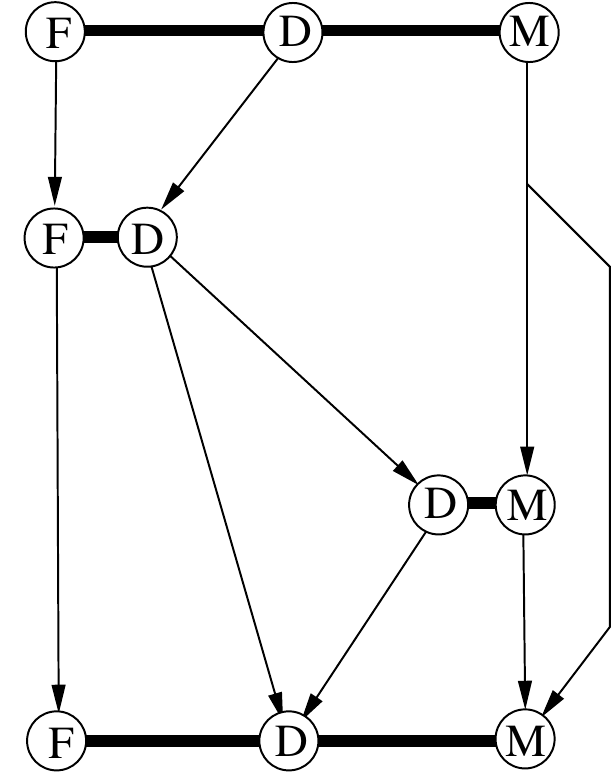} \qquad \scalebox{0.40}{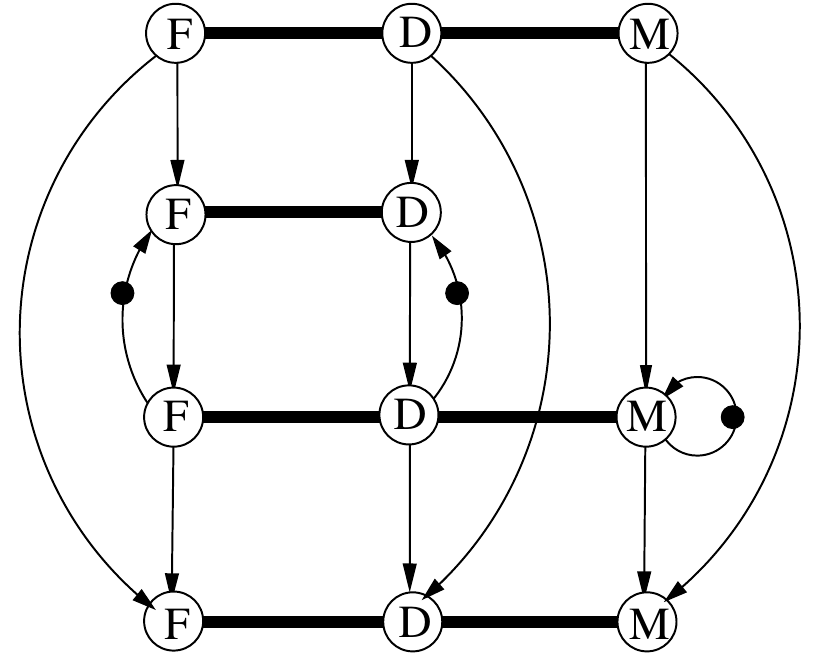}}
\caption{Acyclic and cyclic negotiations.}
\label{fig:dneg}
\end{figure}
Negotiations are graphically represented as shown in Figure \ref{fig:dneg}. 
For each atom $n \in N$ we draw a black bar; for each party $a$ of $P_n$ we 
draw a white circle on the bar, called a {\em port}. For each $(n,a,r) \in T(N)$, 
we draw a hyperarc leading from the port of $a$ in $n$ to all the ports of $a$ in 
the atoms of $\next(n,a,r)$, and label it by $r$. 
Figure \ref{fig:dneg} shows two Father-Daughter-Mother negotiations. On the left,
Daughter and Father negotiate with possible outcomes \texttt{yes} ($\texttt{y}$), 
\texttt{no} ($\texttt{n}$),
and \texttt{ask\_mother} ($\texttt{am}$). If the outcome is the latter, then Daughter and Mother negotiate
with outcomes \texttt{yes}, \texttt{no}. In the negotiation on the right,
Father, Daughter and Mother negotiate with outcomes \texttt{yes} and \texttt{no}.
If the outcome is \texttt{yes}, then Father and Daughter negotiate a return time  
(atom $n_1$) and propose it to Mother (atom $n_2$). 
If Mother approves (outcome \texttt{yes}) , then the negotiation terminates, 
otherwise (outcome \texttt{r}) Daughter and Father renegotiate the return time. 
For the sake of brevity we do not describe the transformers of the atoms.
 

\begin{definition}
The {\em graph associated to a negotiation ${\cal N}=(N, n_0, n_f, \next)$}
is the directed graph with vertices $N$ and edges $\{(n,n')\in N \times N \mid \exists \, (n,a,r) \in T(N) \colon n' \in \next(n,a,r)\} $. The negotiation ${\cal N}$ is {\em acyclic} if its graph has no cycles, {\jde otherwise it is cyclic}.
\end{definition}  

The negotiation on the left of Figure \ref{fig:dneg} is acyclic, the one the right is
cyclic.

\subsection{Semantics}
A {\em marking} of a negotiation ${\cal N}=(N, n_0, n_f, \next)$ is a mapping 
$\vx \colon \agents \rightarrow 2^N$. Intuitively, $\vx(a)$ is the set of atoms that agent $a$ is currently ready to engage in next. 
The {\em initial} and {\em final} markings,  denoted by $\vx_0$ and $\vx_f$ respectively, are given by $\vx_0(a)=\{n_0\}$ and 
$\vx_f(a)=\emptyset$ for every $a \in \agents$.

A marking $\vx$ {\em enables} an atom $n$ if $n \in \vx(a)$ for every $a \in P_n$,
i.e., if every {\jde party of} $n$ is currently ready to engage in it.
If $\vx$ enables $n$, then $n$ can take place and its parties
agree on an outcome $r$; we say that $(n,r)$ {\em occurs}.
Abusing language, we will call this pair also an outcome.
The occurrence of $(n,r)$ produces a next marking $\vx'$ given by $\vx'(a) = \next(n,a,r)$ for every $a \in P_n$, 
and $\vx'(a)=\vx(a)$ for every $a \in \agents \setminus P_n$. 
We write $\vx \by{(n,r)} \vx'$ to denote this,
and call it a {\em small step}. 

{\jde By this definition, $\vx (a)$ is always either $\{n_0\}$ or equals $\next(n,a,r)$ for some atom $n$ and outcome $r$. 
The marking $\vx_f$ can only be reached by the occurrence of $(n_f, r)$ ($r$ being a possible outcome of $n_f$), 
and it does not enable any atom. Any other marking that does not enable any atom is considered a {\em deadlock}.}

Reachable markings are graphically represented by placing tokens {\jd (black dots)} on the forking points of the hyperarcs (or in the middle of an arc). Figure \ref{fig:dneg} shows on the right a marking in {\jde which \texttt{F} and \texttt{D} are ready to engage in $n_1$ and \texttt{M} is ready to engage in 
$n_2$.} 

We write $\vx_1 \by{\sigma}$ to denote that there is a sequence 
$$\vx_1 \by{(n_1,r_1)} \vx_2 \by{(n_2,r_2)}\cdots \by{(n_{k-1},r_{k-1})} \vx_{k} \by{(n_k,r_k)} \vx_{k+1} \cdots$$ 
of small steps such that  $\sigma = (n_1, r_1) \ldots (n_{k}, r_{k}) \ldots$. If
$\vx_1 \by{\sigma}$, then $\sigma$ is an {\em occurrence sequence} from the marking $\vx_1$, and $\vx_1$ enables $\sigma$.
If $\sigma$ is finite, then we write
$\vx_1 \by{\sigma} \vx_{k+1}$ and say that $\vx_{k+1}$ is {\em reachable} from $\vx_1$. 
If  $\vx_1$ is the initial marking then we call $\sigma$ {\em initial occurrence sequence}. If moreover $\vx_{k+1}$ is the final marking, then $\sigma$ is a {\em large step}.

\paragraph{Negotiations and Petri nets.} A negotiation can be assoicated an equivalent 
Petri net with the same occurrence sequences (see \cite{negI}, arXiv version). However,
in the worst case the Petri net is exponentially larger.

\subsection{Soundness}

Following \cite{aalst,DBLP:journals/fac/AalstHHSVVW11}, we introduce a 
notion of well-formedness of a negotiation: 

\begin{definition}
A negotiation is {\em sound} if {\em (a)} every atom is enabled at some reachable marking, and {\em (b)} every occurrence sequence from the initial marking is either a large step or can be extended to a large step. 
\end{definition}

The negotiations of Figure \ref{fig:dneg} are sound. However, if we set {\jde in the left negotiation}
$\next(n_0,\texttt{M}, \texttt{st})= \{n_\texttt{DM}\}$ instead of $\next(n_0,\texttt{M}, \texttt{st})= \{n_\texttt{DM}, n_f\}$, then the occurrence sequence $(n_0,\texttt{st}) (n_\texttt{FD}, \texttt{yes})$
leads to a deadlock.

\begin{definition}
Given a negotiation ${\cal N}=(N,n_0,n_f,\next)$, we attach to each outcome $r$ of $n_f$ a 
{\em summary transformer} $\trf{{\cal N},r}$ as follows. Let $E_r$ be the set of large steps 
of ${\cal N}$ that end with $(n_f,r)$. {\jde We define $\trf{{\cal N},r} = \bigcup_{\sigma \in E_r} \trf{\sigma}$, where for $\sigma = (n_1, r_1)
\ldots (n_k, r_k)$ we define $\trf{\sigma}= \delta_{n_1} (r_1) \cdots \delta_{n_k} (r_k)
$ {\jd (each $\delta_{n_i} (r_i)$ is a relation on $Q_A$;  concatenation is the usual concatenation of relations)}}.
\end{definition}

$\trf{{\cal N},r}(q_0)$ is the set of possible final states of the 
agents after the negotiation concludes with outcome $r$, if their 
initial states are given by $q_0$. 

\begin{definition}
Two negotiations ${\cal N}_1$ and ${\cal N}_2$ over the same set of agents are {\em equivalent}
if they are either both unsound, or if they are both sound, have the same final outcomes (outcomes of the final atom), and 
$\trf{{\cal N}_1, r} = \trf{{\cal N}_2, r}$ for every final outcome $r$.
If ${\cal N}_1$ are equivalent and ${\cal N}_2$ and ${\cal N}_2$ {\jd consists of a single} atom, 
then ${\cal N}_2$ is the {\em summary} of ${\cal N}_1$. 
\end{definition}

Notice that, according to this definition, all unsound negotiations are equivalent. This amounts
to considering soundness essential for a negotiation: if it fails, we do not care about the rest. 

\section{Deterministic Negotiations}
\label{sec:determ}

We introduce 
deterministic negotiations.

\begin{definition}
    
A negotiation $\N$ is deterministic if for every $(n,a,r) \in T(N)$ there is an atom $n'$ such that
$\next(n,a,r) = \{n'\}$
\end{definition}
In the rest of the paper we write $\next(n,a,r) = n'$ instead of $\next(n,a,r) = \{n'\}$.

Graphically, 
a negotiation is deterministic if there are no proper hyperarcs.
The negotiation on the left of Figure \ref{fig:dneg} is not deterministic (it contains a proper hyperarc for Mother), while the one on the right is deterministic. In the sequel, we 
often assume that a negotiation is sound and deterministic, and 
abbreviate ``sound and deterministic negotiation'' to SDN.

\section{Reduction Rules for Deterministic Negotiations}
\label{sec:redrules}

We present three equivalence-preserving reduction rules for negotiations. Two 
of them were already introduced in \cite{negI}, while the iteration rule is new.

A {\em reduction rule}, or just a rule, 
is a binary relation on the set of negotiations. Given a rule $R$,
we write ${\cal N}_1 \by{R} {\cal N}_2$ for $({\cal N}_1, {\cal N}_2) \in R$.
A rule $R$ is {\em correct} if it preserves equivalence, i.e., if 
${\cal N}_1 \by{R} {\cal N}_2$ implies ${\cal N}_1 \equiv{\cal N}_2$. In particular, this 
implies that ${\cal N}_1$ is sound if{}f
${\cal N}_2$ is sound.

Given a set of
rules ${\cal R} = \{R_1, \ldots, R_k\}$, we denote by ${\cal R}^*$ the reflexive 
and transitive closure of $R_1 \cup \ldots \cup R_k$. We say that ${\cal R}$
is {\em complete with respect to a class of negotiations} if, for every
negotiation ${\cal N}$ in the class, there is a negotiation ${\cal N'}$ consisting of a single atom
such that ${\cal N} \by{{\cal R}^*} {\cal N'}$. We describe rules as pairs of a {\em guard} and an {\em action}; 
${\cal N}_1 \by{R} {\cal N}_2$ holds if ${\cal N}_1$ satisfies the guard and 
${\cal N}_2$ is a possible result of applying the action to ${\cal N}_1$.

Slightly more general versions of the following rules have been presented in \cite{negI}. Here we only consider 
deterministic negotiations.\\

\noindent
{\bf \em Merge rule.} Intuitively, the {\em merge rule} merges two outcomes with identical {\jd next enabled atoms} into one single outcome. 

\begin{definition}{Merge rule}

\noindent {\bf Guard}: \begin{tabular}[t]{l} $N$ contains an atom $n$ 
with two distinct outcomes $r_1, r_2 \in \outc_n$
 such \\ that $\next(n,a,r_1) = \next(n,a,r_2)$ for every $a \in \agents_n$.
\end{tabular}

\noindent {\bf Action}: \begin{tabular}[t]{ll}
(1) & $\outc_n \leftarrow (\outc_n \setminus \{r_1, r_2\}) \cup \{r_f\}$,
where $r_f$ is a fresh name. \\
(2) & For all $a \in P_n$: $\next(n,a,r_f) \leftarrow \next(n,a,r_1)$. \\
(3) & $\delta (n, r_f) \leftarrow \delta (n, r_1) \cup \delta (n, r_2)$.
\end{tabular}
\end{definition}

\noindent {\bf \em Shortcut rule.} Inituitively, the shortcut rule merges the outcomes 
of two atoms that can occur one after the other into one single outcome with the same effect. 
Figure \ref{fig:exits} illustrates the definition (ignore the big circle for the moment): 
the outcome $(n, r_f)$, shown in red, is the ``shortcut'' of the outcome $(n,r)$ followed by
the outcome $(n',r')$.

\je{\begin{definition}
Given atoms $n,n'$, we say that $(n,r)$ {\em unconditionally enables} $n'$
if $P_n \supseteq P_{n'}$ and $\next(n,a,r) = n'$ for every $a \in P_{n'}$.
\end{definition}}
\noindent Observe that if $(n,r)$ unconditionally enables $n'$
then, for {\em every} marking $\vx$ that enables $n$, 
the marking $\vx'$ given by $\vx \by{(n,r)} \vx'$ enables $n'$. \je{Moreover, $n'$ 
can only be disabled by its own occurrence.}

\begin{definition}{Shortcut rule for deterministic negotiations}
\label{def:shortcutrule}

\noindent {\bf Guard}: $N$ contains an atom $n$ with an outcome $r$, and 
an atom $n'$, $n' \neq n$, such that $(n,r)$ unconditionally enables 
$n'$.

\noindent {\bf Action}: 
\begin{tabular}[t]{ll}
(1) & $\outc_n \leftarrow (\outc_n \setminus \{r\}) \cup \{r'_f \mid r' \in \outc_{n'}\}$, where $r'_f$ are fresh names. \\
(2) & For all $a \in P_{n'}$, $r' \in \outc_{n'}$: $\next(n,a,r_f') \leftarrow \next(n',a ,r')$.\\
    & For all $a \in P \setminus P_{n'}$, $r' \in \outc_{n'}$: $\next(n,a,r'_f) \leftarrow \next(n,a,r)$. \\
(3) & For all $r' \in \outc_{n'}$: 
$\delta_n (r_f') \leftarrow \delta_n (r) \delta_{n'} (r')$.\\
(4) & If $\next^{-1}(n')=\emptyset$ after (1)-(3), then remove $n'$ from $N$, where \\
    & $\next^{-1}(n') = \{ (\tilde{n},\tilde{a},\tilde{r}) \in T(N) \mid n' \in \next(\tilde{n},\tilde{a},\tilde{r}) \}$.
\end{tabular}
\end{definition}


\noindent {\bf \em Iteration rule.} Loosely speaking, the iteration rule replaces the iteration of a negotiation by one single atom
with the same effect.

\begin{definition}{Iteration rule}

\noindent {\bf Guard:} $N$ contains an atom $n$ with an outcome $r$ such that 
$\next(n,a,r)= n$ for every party $a$ of $n$.

\noindent {\bf Action:} 
\begin{tabular}[t]{ll}
(1) & $R_n \leftarrow \{r_f' \mid r' \in R_n \setminus \{r\}\}$. \\
(2) & For every $r_f' \in R_n $: 
$\delta_n (r_f') \leftarrow \delta_n (r)^*\: \delta_n (r')$.
\end{tabular}

\end{definition}


It is important to notice that reductions preserve determinism:
\begin{proposition}
If a negotiation $\N$ is deterministic and the application of the shortcut, merge or iteration rule yields negotiation $\N'$ then $\N'$ is deterministic, too.
\end{proposition}
\begin{theorem}
The merge, shortcut, and iteration rules are correct.
\end{theorem}
\begin{proof}
Correctness of the merge and iteration rules is obvious. The correctness of a more general version of the shortcut rule is proved in \cite{negI}\footnote{ The rule of
\cite{negI} has an additional condition in the guard which is always true for deterministic negotiations.}.\qed
\end{proof}

\section{Completeness}
\label{sec:compl}

In \cite{negI} we show that every sound and weakly deterministic acyclic negotiation can be 
summarized to a single atom, and that in the deterministic case the number of rule applications 
is polynomial (actually, \cite{negI} provides a sharper bound than the one in this
 theorem):

\begin{theorem}[\cite{negI}]
\label{thm:polcomp}
Every sound deterministic acyclic negotiation $\cal N$ can be reduced to a single atom by means of 
$|N|^2 + |\Out({\cal N})|$ applications of the merge and shortcut rules, where $N$ is the set of atoms 
of $\N$, and $\Out({\cal N})$ is the set of all outcomes of all atoms of $N$.
\end{theorem}

In the rest of the paper section we prove that, surprisingly, the addition of the very simple iteration rule suffices to extend this result to cyclic deterministic negotiations, although with a higher exponent. The argument is complex, and requires a
detailed analysis of the structure of SDNs.

In this section we present the completeness proof, while the complexity result is presented in the next. We illustrate the reduction algorithm by means of an example. Figure \ref{fig:red} (a) shows a cyclic SDN similar to the Father-Daughter-Mother
negotiation on the right of Figure \ref{fig:dneg}. We identify an ``almost acyclic'' fragment, namely the fragment coloured blue in the figure. 
Intuitively, ``almost acyclic'' means that the 
fragment can be obtained by ``merging'' the initial and final atoms of an acyclic SDN; in our example, this is the blue acyclic SDN shown in Figure \ref{fig:red} (b). This acyclic SDN can be summarized using the shortcut and merge rules.  If we apply the same sequence of rules to the blue fragment (with the exception of the last rule, which reduces a negotiation with 
two different atoms and one single outcome to an atomic negotiation) we obtain the negotiation shown in (c). 
The blue self-loop can now be eliminated with the help of the iteration rule, and the procedure can be iterated: We identify an ``almost acyclic'' fragment, coloured red. Its reduction yields the the negotiation shown in (e). The self-loop is eliminated by the iteration rule, yielding an acyclic negotiation, which can be summarized. 

\begin{figure}[h]
\centerline{\scalebox{0.35}{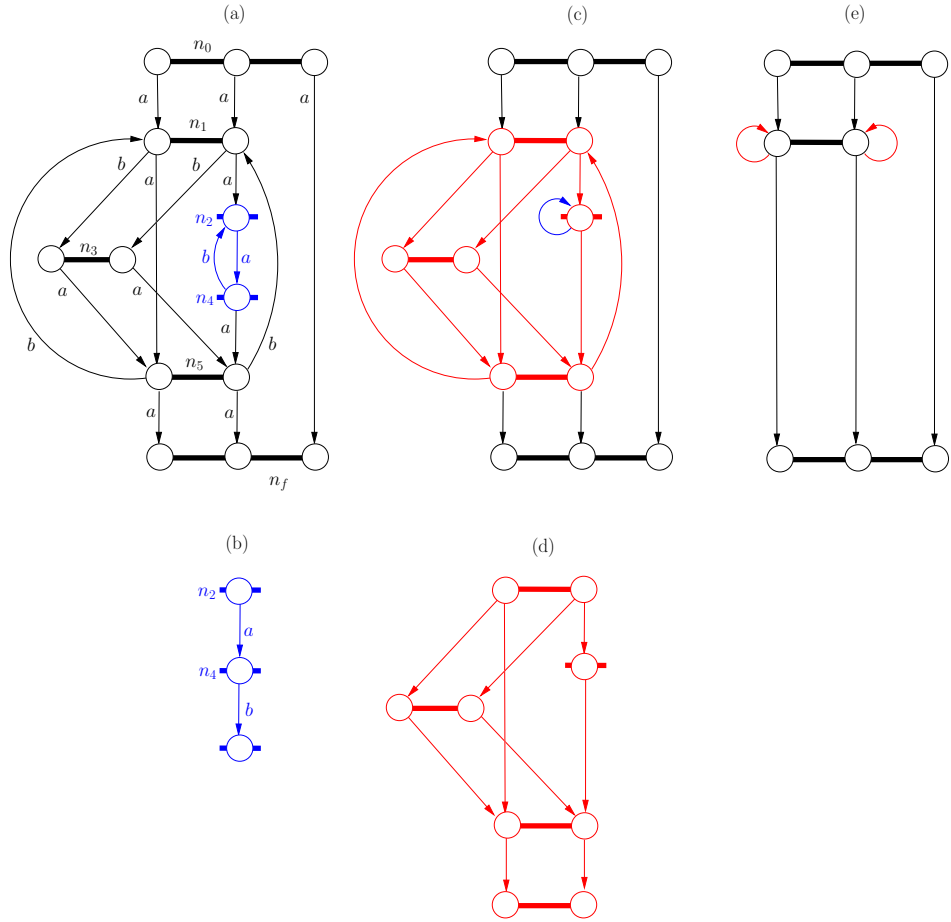}}
\caption{The reduction procedure}
\label{fig:red}
\end{figure}

In order to prove completeness  we must show that every cyclic SDN contains at least one almost acyclic fragment, which is non-trivial. The proof has three parts: We first show that every cyclic SDN has a  {\em loop}: an occurrence sequence from some reachable marking $\vx$ back to $\vx$. Then we show that each minimal loop has a {\em synchronizer}: an atom involving  each agent that is party of any atom of the loop. Finally we show how to use synchronizers to identify a nonempty and almost acyclic fragment. 

\subsection{Lassos and Loops}
\label{subsec:lassos}

\begin{definition}
A {\em lasso} of a negotiation is a pair $(\rho,\sigma)$ of occurrence sequences such that $\sigma$ is not the empty sequence and
$\vx_0 \by{\rho} \vx \by{\sigma} \vx$ for some marking $\vx$. 
A {\em loop} is an occurrence sequence $\sigma$ such that $(\rho, \sigma)$
is a lasso for some occurrence sequence $\rho$.
A {\em minimal loop} is a loop $\sigma$ satisfying the property that there is no other loop $\sigma'$ such that the set of atoms in $\sigma'$ is a proper subset of the set of atoms in $\sigma$.
\end{definition}

Observe that lassos and loops are behavioural notions, i.e., structures of the reachability graph of a negotiation. The following result establishes relations between loops and cycles, where cycles are defined on the graph of a negotiation.

\begin{lemma}
\label{lem:struct}
\begin{itemize}
\item[(1)] Every cyclic SDN has a loop.
\item[(2)] The set of atoms of a minimal loop generates a strongly 
connected subgraph of the graph of the considered negotiation.
\end{itemize}
\end{lemma}
\begin{proof} 
\begin{itemize}
\item[(1)]
Let $\pi$ be a cycle of the graph of the negotiation ${\cal N}$. Let $n_1$ be an arbitrary atom occurring in $\pi$, and let $n_2$ be its successor in $\pi$. We have $n_1 \neq n_f$ because $n_f$ has no successor, and hence no cycle contains $n_f$.

By soundness, some reachable marking $\vx_1$ enables $n_1$. There is an agent $a$ and a result $r$
such that $\next (n_1,a,r)$ contains $n_2$. By determinism we have $\next (n_1,a,r)=\{n_2\}$.
Let $\vx_1 \by{(n_1,r)} \vx_1'$. 
Again by soundness, there is an occurrence sequence from $\vx_1'$ that leads to the final marking. This sequence necessarily contains an occurrence of $n_2$ because 
this is the only atom  agent $a$ is ready to engage in. In particular, some prefix of this sequence leads to a marking $\vx_2$ that enables $n_2$. 

Repeating this argument for all nodes $n_1$, $n_2$, $n_3$, \ldots , $n_k = n_1$ of the cycle $\pi$, we conclude that there is an infinite occurrence sequence, containing infinitely many occurrences of atoms of the cycle $\pi$. Since the set of reachable markings is finite, this sequence contains a loop. 
\item[(2)]
For each agent involved in any atom of the loop, consider the sequence of atoms 
this agent is involved in. By the definition of the graph of the negotiation, 
this sequence is a path of the graph. It is moreover a (not necessarily simple) cycle
of teh graph, because a loop starts and ends with the same marking. So the subgraph generated 
by the atoms in the loop is covered by cycles. It is moreover strongly connected because, for each 
proper strongly connected component, the projection of the atoms of the loop onto the atoms 
in the component is a smaller loop, contradicting the minimality of the loop. \qed
\end{itemize}
\end{proof}

\subsection{Synchronizers}
\label{subsec:synchronizers}

\begin{definition}
A loop $\sigma = (n_1, r_1) \ldots (n_k, r_k)$ is {\em synchronized} if there is an atom $n_i$ in $\sigma$ such that $P_j \subseteq P_i$ for every 
$1 \leq j \leq k$, i.e., every party of every  atom in the loop is also a party of $n_i$. 
We call $n_i$ a {\em synchronizer} of the loop. An atom is a synchronizer of a negotiation if it is a synchronizer of at least one of its loops.
\end{definition} 

Observe that each loop $\vx \by{(n,r)} \vx$ is synchronized. In the graph associated to a negotiation, such a loop appears as a self-loop, i.e., as an edge from atom $n$ to atom $n$.

Some of the loops of the SDN shown in Figure \ref{fig:red} (a) are $(n_1,a) \, (n_2,a) \, (n_4, a) \, (n_5, b)$,  $(n_1,b) \, (n_3,a) \, (n_5, b)$, and $(n_2, a) \,(n_4, b) $.
The first loop is synchronized by $(n_1,a)$ and by $(n_5,b)$, the two others are synchronized by all their outcomes. 

The main result of this paper is strongly based on the following lemma.

\begin{lemma}
\label{lem:struct2}
Every minimal loop of a SDN is synchronized. 
\end{lemma}
\begin{proof}
Let $\sigma$ be a minimal loop, enabled at a reachable marking $\vx$.
Define $N_\sigma$ as the set of atoms that occur in $\sigma$
and $A_\sigma$ as the set of agents involved in atoms of $N_\sigma$.
Since $\N$ is sound, there is an occurrence sequence $\sigma_f$ enabled by $\vx$ that ends with the final atom $n_f$. 

Now choose  an arbitrary agent $\hat{a}$ of $A_\sigma$. Using $\sigma_f$, we
construct a path $\pi$ of the graph of $\N$ as follows:
We begin this path with the last atom $n \in N_\sigma$ that appears in $\sigma_f$ and involves agent $\hat{a}$. We call this atom $n_\pi$. Then we
repeatedly choose the last atom in $\sigma_f$ that involves $\hat{a}$ and moreover is a successor of the last vertex of the path constructed so far. 
By construction, this path has no cycles (i.e., all vertices are distinct), starts with an atom of $N_\sigma$ and has not further atoms of $N_\sigma$, ends with $n_f$, and only contains atoms involving $\hat{a}$.

Since $\vx$ enables the loop $\sigma$ and since $n_\pi \in N_\sigma$, after some prefix of $\sigma$ a marking $\vx_\pi$ is reached which enables $n_\pi$. The loop
$\sigma$ continues with some outcome $(n_\pi,r_1)$, where $r_1$ is one possible result of $n_\pi$. 

By construction of the path $\pi$, 
there is an alternative result $r_2$ of $n_\pi$ such that $\next (n_\pi,\hat{a},r_2)$ is the second atom of the path $\pi$, and this atom does not belong to $N_\sigma$. 
Let $\vx_\pi'$ be the marking reached after the occurrence of $(n_\pi,r_2)$ at
 $\vx_\pi$.

From $\vx_\pi'$,  we iteratively construct an occurrence sequence as follows:
\begin{itemize}
\item[(1)] if an atom $n$ of $N_\sigma$ is enabled and thus some $(n,r)$ occurs in $\sigma$, we continue with $(n,r)$,
\item[(2)] otherwise, if an atom $n$ of the path $\pi$ is enabled, we let this atom occur with an outcome $r$ such that $\next (n,\hat{a},r)$ is the successor atom w.r.t.\  the path $\pi$,
\item[(3)] otherwise we add a minimal occurrence sequence that either leads to the final marking or enables an atom of $\sigma$ or an atom of $\pi$, so that after this sequence one of the previous rules can be applied. Such an occurrence sequence exists because $\N$ is sound and hence the final marking can be reached.
\end{itemize}

First observe that agent $\hat{a}$ will always be ready to engage only in an atom of the path $\pi$. So its token is moved along  $\pi$. Conversely, all atoms of $\pi$ involve $\hat{a}$. Therefore only finitely many atoms of $\pi$ occur in the sequence. This limits the total number of occurrences of type (2).

Agent $\hat{a}$ is no more ready to engage in any atom of $N_\sigma$ during the sequence. So at least $n_\pi$ cannot occur any more in the sequence because $\hat{a}$ is a party of $n_\pi$. By minimality of the loop $\sigma$, there is no loop with a set of atoms in $N_\sigma \setminus \{n_\pi\}$. Since the set of reachable markings is finite, there cannot be an infinite sequence of atoms of $N_\sigma$ (type (1)) without occurrences of other atoms. 

By determinism, each agent ready to engage in an atom of $N_\sigma$ 
can only engage in this atom. So the set of these agents is only changed by occurrences of type (1). By construction, no agent ever leaves the loop after the occurrence of $(n_\pi,r_2)$, i.e.\ every agent of this set remains in this set by an occurrence of type (1). Therefore, the set of agents ready to engage in an atom of $N_\sigma$ never decreases.

For each sequence of type (3) we have three possibilities. 
\begin{itemize}
\item[(a)]
It  ends with the final marking. 
\item[(b)]
It ends with a marking that enables an atom of (2), which then occurs next. However, atoms of (2)  can  occur only finitely often in the constructed sequence, as already mentioned.
\item[(c)]
It ends  with a marking that enables an atom of (1) which then occurs next. In that case the last outcome of this sequence necessarily involves an agent of $A_\sigma$, which after this occurrence is ready to engage in an atom of $N_\sigma$. So it increases the number of agents ready to engage in an atom of $N_\sigma$. Since this number never decreases, this option can also happen only finitely often.
\end{itemize}

Hence, eventually only option (a) is possible, and the sequence  will reach the final marking.
Since the final atom involves all agents, no agent was able to remain in the loop.
In other words: all agents of $A_\sigma$ left the loop when $(n_\pi, r_2)$ has occurred.
As a consequence, all these agents are parties of $n_\pi$, and $n_\pi$ therefore is a synchronizer of the loop $\sigma$.
\qed
\end{proof}

Observe that this lemma does not hold for
arbitrary (i.e., non-deterministic) sound negotiations. 
For the negotiation on the right of
Figure \ref{fig:cyclicnoloops} (all atoms have only one outcome, whose name is omitted),
the sequence $n_1 \, n_2$ is a loop without synchronizers.

The negotiation on the left shows that Lemma \ref{lem:struct}(1) also holds only in the deterministic case. It is sound and cyclic, but has no loops, because the only big step
is $n_0 \, n_1 \, n_2 \, n_1 \, n_f$ (the name of the outcome is again omitted). 

\begin{figure}[h]
\centerline{\scalebox{0.35}{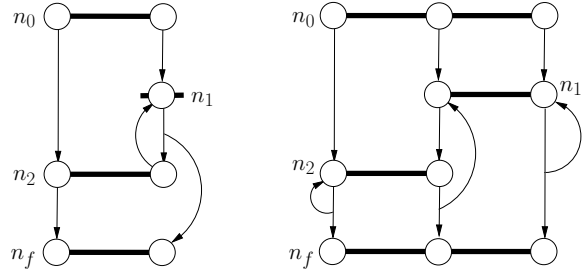}}
\caption{Two sound and cyclic negotiations}
\label{fig:cyclicnoloops}
\end{figure}

\subsection{Fragments}
\label{subsec:fragments}

We assign to each atom $n$ of an SDN a ``fragment'' $\F_n$ as follows: 
we take all the loops synchronized by $n$, and (informally) define $\F_n$ 
as the atoms and outcomes that appear in these loops. Figure \ref{fig:cyclicfrag} (a) and (c) show 
$\F_{n_1}$ and $\F_{n_2}$ for the SDN of Figure \ref{fig:red}. Since a cyclic SDN has at least one 
loop and hence also a minimal one, and since every loop has a synchronizer, at least one of the 
fragments of a cyclic SDN is nonempty. 

Given a fragment $\F_{n}$, let $\N_{n}$ denote the negotiation obtained
by, intuitively, ``splitting'' the atom $n$ into an initial and a final atom. 
Figure \ref{fig:cyclicfrag} (b) and (d) show the ``splittings'' $\N_{n_1}$ and $\N_{n_2}$ of 
$\F_{n_1}$ and $\F_{n_2}$. Not all fragments are almost acyclic. For instance, $\N_{n_1}$ 
is not acyclic, and so $\F_{n_1}$ is not almost acyclic. However, we 
prove that if a fragment is not almost acyclic, then it contains a smaller
fragment (for instance, $\F_{n_1}$ contains $\F_{n_2}$). This shows that
every minimal fragment is almost acyclic.

\begin{figure}[h]
\centerline{\scalebox{0.35}{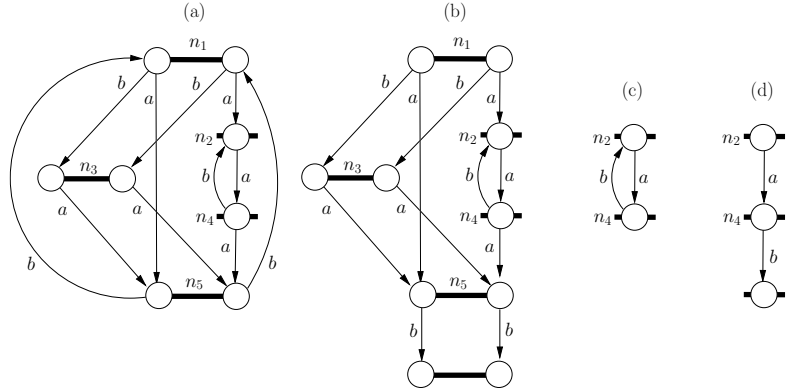}}
\caption{Fragments of the SDN of Figure \ref{fig:red}(a) and their ``splittings''}
\label{fig:cyclicfrag}
\end{figure}

\begin{definition}
Let ${\cal L}$ be a set of loops of $\N$. 
Abusing language, we write $(n, r) \in {\cal L}$ resp. $n \in {\cal L}$ to denote that $(n,r)$ resp. $n$ appears in some loop of ${\cal L}$. The {\em projection} of an atom 
$n=(P_n,R_n,\delta_n) \in {\cal L}$ onto ${\cal L}$ is the atom $n_\L = (P_\L, R_\L,\delta_\L)$, where $P_\L = P_n$, $R_\L = \{  r \mid (n,r) \in {\cal L}\}$, and $\delta_\L((n_\L, r)) = \delta((n,r))$ for every $(n,r) \in  {\cal L}$.
\end{definition}

\begin{definition}
\label{def:fragment}
Let $s$ be an atom of a negotiation $\N$, and let $\L$ be the set of loops 
synchronized by $s$. The {\em $s$-fragment} of $\N$ is the pair $\F_s=(F_s, \next_s)$, where 
$F_s = \{ n_\L \mid n \in \L \}$ and $\next_s(n_\L,a,r)=\next(n,a,r)$ for every
$a \in P_\L$ and $r \in R_\L$.

The {\em $s$-negotiation} of $\N$ is the negotiation  $\N_s= (N_s, n_{s0}, n_{sf}, \next'_s)$,
where 
\begin{itemize}
\item $N_s$ contains the atoms of $F_s$ plus a fresh atom $n_{sf}$; 
\item $n_{s0} = s_\L$; and
\item For every $n_\L \in F_s$, $a \in P_\L$, and $r \in R_\L$:
$$\next_s'(n_\L,a,r) = \left\{
\begin{array}{ll}
\next(n,a,r) & \mbox{ if $\next(n,a,r) \neq s$ }\\
n_{sf}       & \mbox{ otherwise}
\end{array}\right.$$
\end{itemize}
\end{definition}

The following proposition proves some basic properties of $s$-negotiations.

\begin{proposition}\label{prop:occseq}
Let $s$ be an atom of a negotiation $\N$. If $\N$ is a SDN, then  $\N_s$ is a SDN.
\end{proposition}
\begin{proof}
Let $n'$ be an atom of $\N_n$. By definition, there is a lasso $(\rho,\sigma)$ of
$\N$, synchronized by $n$, such that $n'$ appears in $\sigma$. By the definition of $\N_n$, we have
that $\sigma$ is an occurrence sequence of $\N_n$, and so that $n'$ can occur in $\N_n$.

Assume now that $\sigma$ is an occurrence sequence of $\N_n$ that is not a large 
step of $\N_n$ and cannot be extended to a large step of $\N_n$. 
W.l.o.g. we can assume that $\sigma$ does not enable any atom of $\N_n$. 
Let $\vx$ be the marking reached by $\sigma$. It follows easily from the definition
of $\N_n$ that there is an occurrence sequence $\rho \sigma$ of $\N$. 
Moreover, if  $\vx$ be the marking reached by this sequence, then
$\vx'$ is the projection of $\vx$ onto the set $P$ of parties of $n$.

By soundness there is a maximal
occurrence sequence $\rho\rho'$ such that $\rho'$ contains only atoms with parties
in $\agents\setminus P_n$. Clearly $\rho\rho'\sigma$ is an occurrence sequence of $\N$.
Let $\vx''$ be the marking reached by $\rho\rho'\sigma$. Clearly, we still have that
$\vx'$ is the projection of $\vx''$ onto $P$. We claim that $\vx''$ is a deadlock, 
contradicting the soundness of $\N$. To prove the claim, assume that $\vx''$ 
enables some atom $n'$ with set of parties $P'$. If $P' \subseteq P_t$, then 
$n'$ is also enabled at $\vx'$, contradicting that $\sigma$ does not enable any atom of $\N_n$. If $P' \subseteq \agents\setminus P$, then $\rho\rho'$ enables $n'$,
contradicting the maximality of $\rho\rho'$. Finally, if $P' \cap P \neq \emptyset P' \cap (\agents\setminus P)$, then there is an agent $a \in P$ such that $\vx''(a) \notin N_n$. But by definition of $\sigma$ we have $\vx'(a) \in N_n$, contradicting that $\vx'$ is the
projection onto $P$ of $\vx''$. \qed
\end{proof}

\begin{lemma}
A cyclic SDN contains an atom $n$ such that $\N_n$ is an acyclic SDN. 
\label{lem:soundsubneg}
\end{lemma}
\begin{proof}
Let $\N$ be a cyclic SDN. By Lemma \ref{lem:struct}, $\N$ has a loop and hence a also minimal loop. By
Lemma \ref{lem:struct2} this loop has a synchronizer $n$, and so  $\N_n$ is nonempty.
Choose $n$ so that $\N_n$ is nonempty, but its number of atoms is minimal. We claim that
$\N_n$ is acyclic. Assume the contrary. By Lemma \ref{prop:occseq}, $\N_n$ is
a SDN. By Lemmas \ref{lem:struct} and \ref{lem:struct2}, exactly as above, $\N_n$ contains an atom $n'$ such that $\N_{nn'}$ is nonempty. Clearly, 
$\N_{nn'}$ contains fewer atoms than $\N_n$ and is isomorphic to $\N_{n'}$. 
This contradicts the minimality of $\N_n$.\qed
\end{proof}

The example on the left of Figure \ref{fig:cyclicnoloops} shows that this result does not hold 
for the non-deterministic case.  

\subsection{The reduction procedure}
\label{subsec:redproc}
We can now finally formulate a reduction procedure to summarize an arbitrary SDN.
\smallskip
\begin{center}
\begin{minipage}{12cm}
{\bf Input:} a deterministic negotiation $\N_0$;

\smallskip

{\algo
$\N \leftarrow \mbox{result of exhaustively applying the merge rule to $\N_0$}$;
\while{\mbox{$\N$ is cyclic }}
> select $s \in N$ such that $\N_s$ is acyclic;
> apply to $\N$ the sequence of rules used to summarize $\N_s$ (but the last);\label{ac}
> apply the iteration rule to $s$;\label{it}
> exhaustively apply the merge rule \label{me}
apply the reduction sequence of Theorem \ref{thm:polcomp}}
\end{minipage}
\end{center}
\smallskip


\begin{theorem}
\label{thm:completeness}
The reduction procedure returns a summary of $\N_0$ if{}f $\N_0$ is sound.
\end{theorem}
\begin{proof}
By induction on the number $k$ of atoms of $\N$ that synchonize at least one loop.
If $k=0$, then by Lemma \ref{lem:struct} and \ref{lem:struct2} $\N$ is acyclic, and the result follows 
from Theorem \ref{thm:polcomp}. If $k > 0$, then by Lemma \ref{lem:soundsubneg} $\N$ contains an 
almost acylic fragment $\F_s$, and so $\N_s$ is acyclic. Since the sequence of rules of line \ref{ac} 
summarizes $\N_s$, its application to $\N$ ends with a negotiation having a unique self-loop-outcome on 
$s$. After removing this outcome with the iteration rule in line \ref{it}, we obtain a SDN with $k-1$ synchronizers, which can be summarized by induction hypothesis (line \ref{me} is not necessary for completeness,
but required for the complexity result of the next section).
\end{proof}

\section{Complexity} 
\label{sec:complexity}

We analyze the number of rule applications required by the reduction procedure.
Let $\N_{i} = (\N_i, n_{0i}, n_{fi}, \next_i)$  
be the negotiation before the $i$-th execution of the while oop. 
The next lemma collects some basic properties of the sequence $\N_1,\N_2, \ldots$.

\begin{lemma}
\label{lem:basics}
For every $i \geq 1$: 
\begin{itemize}
\item[(a)] $N_{i+1} \subseteq N_i$; 
\item[(b)] the merge rule cannot be applied to $\N_i$; and 
\item[(c)] $\N_{i+1}$ has fewer synchronizers than $\N_i$.\\
In particular, by (c) the while loop is executed at most $|N_1|=|N_0|$ times.
\end{itemize}
\end{lemma}
\begin{proof}
Parts (a) and (b) follow immediately from the definitions of the rules and the 
reduction algorithm. For (c), we observe that every synchronizer of $\N_{i+1}$ 
is a synchronizer of $\N_{i}$, but the atom $s$ selected at the $i$-th loop execution is not a synchronizer of $\N_{i+1}$, because all loops synchronized by $s$ are collapsed to self-loops on $s$ during the $i$-th iteration of the loop, and then removed by the iteration rule.\qed
\end{proof}

By Theorem \ref{thm:polcomp}, during the $i$-th iteration of the while loop
line \ref{ac} requires at most $|N_i|^2 + |\Out(\N_i)|$ rule applications. 
Line \ref{it} only requires one application. Now, let $\N_i'$ be the negotiation obtained after the
execution of line \ref{it}. The number of rule applications of line \ref{me} is 
clearly bounded by the number of outcomes of $\Out(\N_i')$ .
For the total number of rule applications  ${\it Appl}(\N_0)$ we then obtain.

\begin{lemma}\label{lem:firstbound}
$${\it Appl}(\N_0) \in {\cal O} ( \; |N_0|^3 + |N_0| \sum_{i=1}^{|N_0|} |\Out(\N_i)|+|\Out(\N_i')| \;)$$
\end{lemma}
\begin{proof}
$$\begin{array}{rcll}
{\it Appl}(\N_0) & \leq  & \displaystyle \sum_{i=1}^{|N_0|} (|N_i|^2 + |\Out(\N_i)| + 1 + |\Out(\N_i')|) 
& \mbox{ \quad Lemma \ref{lem:basics}(c),}\\
& & & \mbox{ \quad Theorem \ref{thm:polcomp}} \\
& \leq  & \displaystyle \sum_{i=1}^{|N_0|} (|N_0|^2 +1 +|\Out(\N_i)|+|\Out(\N_i')|) 
& \mbox{ \quad Lemma \ref{lem:basics}(a)} \\
& \in  & {\cal O} ( \; |N_0|^3 + |N_0| \sum_{i=1}^{|N_0|} |\Out(\N_i)|+|\Out(\N_i')| \;) & \hfill \qed
\end{array}$$
\end{proof}
However, we cannot yet bound ${\it Appl}(\N_0)$ by a polynomial in
$|N_0|$ and $|\Out(N_0)|$, because, in principle, 
the number of outcomes of $\N_i$ or $\N_i'$ might grow exponentially with $i$. 
Indeed, the shortcut rule can increase the number of outcomes.
Consider the degenerate negotiation $\N$ with only one agent shown in Figure \ref{fig:oneagent}(a). 
\begin{figure}[h]
\centerline{\scalebox{0.30}{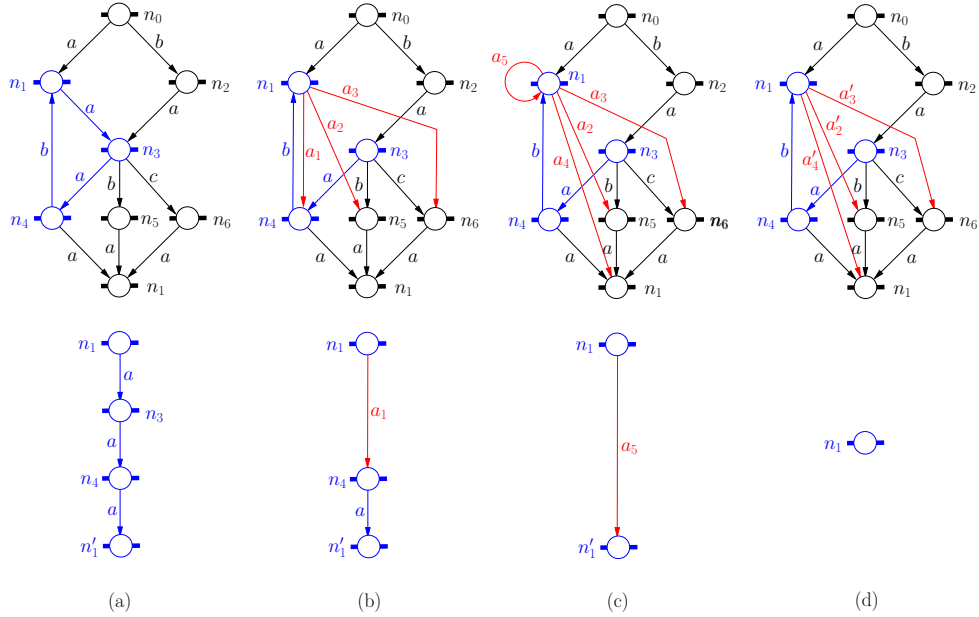}}
\caption{Reducing an SND with one agent}
\label{fig:oneagent}
\end{figure}
${\cal N}$ has one single loop, namely  
$(n_1, a) \, (n_3, a) \, (n_4, b)$. The fragment $\F_{n_1}$ is shown 
in blue, and  ${\cal N}_{n_1}$ is shown below $\N$. 
The negotiation ${\cal N}_{n_1}$ can be summarized by means of three aplications of the shortcut rule,
shown in the lower row of the figure. The upper row
shows the result of application of the same rules to ${\cal N}$. 

The first application removes $n_3$ from ${\cal N}_{n_1}$ 
but not from ${\cal N}$, because $n_3$ has more than one input arc in $\N$
(Figure \ref{fig:oneagent}(b)). Moreover, 
the rule adds three outcomes to $\N$, shown in red. 
The second application removes $n_4$ from ${\cal N}_{n_1}$ 
but not from ${\cal N}$, and adds two new outcomes $(n_1, a_4)$ and $(n_1, a_5)$
(Figure \ref{fig:oneagent}(c)). The third application removes $n_1'$ from ${\cal N}_{n_1}$; in ${\cal N}$ it is replaced by an application of the iteration rule, yielding the negotiation at the top of Figure \ref{fig:oneagent}(d), which has two 
outcomes more than the initial one.

To solve this problem we introduce {\em targets} and {\em exits}.

\subsection{Sources, targets, and exits}
\label{subsec:targets}

\begin{definition}
Let $\N = (N, n_0, n_f, \next)$ be a negotiation, and let $(n,r)$ be an outcome.
The source of $(n, r)$ is $n$. The {\em target} of $(n,r)$ is the 
partial function $\agents \rightarrow N$ that assigns to every party $a \in P_n$ the atom $\next(n,a,r)$, and is undefined for every $a \in \agents \setminus P_n$. The set of {\em targets of $\N$}, denoted by $\Ta(\N)$,
contains the targets of all outcomes of $\N$.
\end{definition}

Consider the reduction process from $\N_i$ to $\N_{i+1}$. It proceeds by 
applying to $\N_i$ the same sequence of rules that summarizes an acyclic  negotiation $\N_s$. 
This sequence progressively reduces the fragment $\F_s$ until it consists of self-loops
on the atom $s$, which can then be reduced by the iteration rule. However, the
sequence also produces new outcomes of $s$ that leave $\F_s$, and which become
outcomes of $\N_{i+1}$ not present in $\N_i$. Consider for instance Figure \ref{fig:exits}(a), which sketches an application of the shorcut rule. The outcome  
$(n, r)$ unconditionally enables $n'$, whose outcome $(n', r')$ makes the left agent leave $\F_s$. The target of $(n, r_f')$ assigns the agents of the negotiations to atoms
$n_1$, $ n_2$ and $ n_3$, respectively. This target is different from the targets of the other atoms in the figure. 

\begin{figure}[htb]
\centerline{\scalebox{0.35}{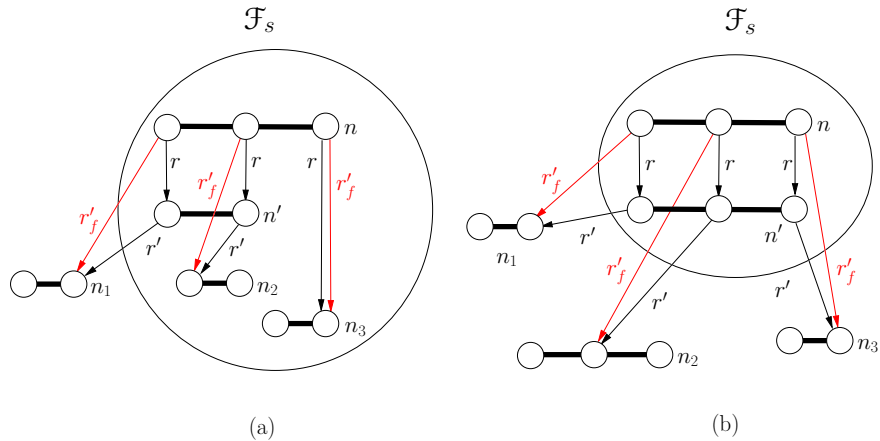}}
\caption{Exits of SNDs}
\label{fig:exits}
\end{figure}

We investigate the sources and targets of outcomes that leave $\F_s$. We call them {\em exits of $\F_s$}.

\begin{definition}
Let $\F_s$ be a fragment of $\N$. An {\em exit} of $\F_s$ is an outcome 
$(n, r) \in \Out(\N)$ such that $n \in F_s$ but $(n, r) \notin \Out(\F_s)$.
\end{definition}

The following lemma presents a key property of the exits of fragments of SDNs:
the occurrence of an exit $(n, r)$ of $\F_s$ forces all agents of $P_s$ to leave the fragment $\F_s$. In other words: all agents of $P_s$ are parties of $n$, and 
the occurrence of $(n,r)$ does not lead any agent back to an atom of $\F_s$.

\begin{lemma}
Let $\F_s$ be a fragment of a SDN $\N$, and let $(e, r_e)$ be an exit of $\F_s$. 
Then $e$ has the same agents as $s$  (i.e., $e$ is also a synchronizer of $\F_s$), 
and $\next(e,a,r_e) \notin F_s$ for every agent $a$ of $e$. 
\label{lem:exit}
\end{lemma}
\begin{proof}
We proceed indirectly and assume that either $P_e \subset P_s$ 
($P_e \subseteq P_s$ by the definition of fragment) or 
$\next(e,a,r_e) \in F_s$ for some $a\in P_e$. Then
at least one agent $h \in P_s$ satisfies 
either $h \notin P_e$ or $\next(e,h,r_e) \in F_s$. We call 
$h$ a {\em home agent} (intuitively, an agent that does not 
leave ``home'', i.e., $\F_s$,  by the occurrence of the exit). 
We show that the existence of $h$ leads to a contradiction.

We partition the set of $\agents$ of agents into internal agents, the agents
of $P_s$, and external agents, the agents of $\agents \setminus P_s$. 
We also partition the set of atoms: an atom is internal if it has only 
internal parties, otherwise it is external. Clearly all atoms of $F_s$ 
are internal, but there can also be internal atoms outside $F_s$. 
If $P_s$ contains all agents of the negotiation, then all agents are internal,
and so are all atoms (also the final atom $n_f$). 
Otherwise at least $n_f$ has an external party and is hence an external atom.

Next we define a function $p \colon N \to Out (\N)$ that assigns to each atom one of 
its outcomes (the {\em preferred} outcome). $p$ is defined for internal and external 
atoms separately, i.e., it is the union of functions $p_i$ 
assigning outcomes to internal atoms, and $p_e$ assigning outcomes to external atoms.

If there are external atoms, and hence $n_f$ is external, $p_e$ is defined as follows. 
First we set $p_e(n_f)$ to an arbitrary outcome of $n_f$. 
Then we proceed iteratively: If some external atom $n$ has an outcome $r$ and an external agent 
$a$ such that $p_e(\next(n,a,r))$ is defined, then set $p_e(n):=r$ (if there are several possibilities, 
we choose one of them arbitrarily).
At the end of the procedure $p_e$ is defined for every external atom, 
because each external atom $n$ has an external agent, say $a$, and, 
since $a$ participates in $n_f$,
the graph of $\N$ has a path of atoms, all of them with $a$ as party, 
leading from $n$
to $n_f$. 

Now we define $p_i$ for internal atoms. 
For the internal atoms $n$ not in $\F_s$ we define $p_i(n)$
arbitrarily. For the internal atoms $n \in F_s$ such that $\next(n,a,r)=s$ 
for some agent $a$ we set $p_i(n)=r$. For the rest of the internal atoms of $F_s$
we proceed iteratively. If $n \in F_s$ has an outcome $r$ and an agent $a$ 
(necessarily internal) such that $p_i(\next(n,a,r))$ is defined, 
then we set $p_i(n,a):=r$ (if there are several 
possibilities, we choose one of them). By Lemma \ref{lem:struct2}(2), the graph
of $\F_s$ is strongly connected, and so eventually $p_i$ is defined for 
all atoms of $\F_s$.

Let $\sigma$ be an arbitrary occurrence sequence leading to a marking $\vx_s$
that enables $s$ (remember that $\N$ is sound). By the definition
of $\F_s$, the marking $\vx_s$ enables an occurrence sequence $\sigma_e$ 
that starts with an occurrence of $s$, contains only atoms of
$F_s$, and ends with an occurrence of $(e, r_e)$, the considered 
exit of $\F_s$.

We now define a  {\em maximal} occurrence sequence $\tau$ enabled at $\vx_s$.
We start with $\tau:=\epsilon$ and while $\tau$ enables some atom
proceed iteratively as follows:
\begin{itemize}
\item If $\tau$ enables $\sigma_e$, then $\tau:=\tau\sigma_e$, i.e., we 
extend the current sequence with $\sigma_e$.
\item Otherwise, choose 
any enabled atom $n$, and set $\tau:= \tau (n,p(n))$, i.e., 
we extend the current sequence with  $(n,p(n)$.
\end{itemize} 

We first show that $\tau$ is infinite, i.e., that we never exit the while loop.
By soundness, there is always an 
enabled atom as long as the final marking is not reached, i.e., as long 
as at least one agent is ready to engage in an atom. 
So it suffices to show that this is the case. We prove that the home agent 
$h$ is ready to engage in an atom after the occurrence of an arbitrary finite prefix of $\tau$.
This result follows from the following claim.\\

\noindent
{\em Claim.}
If $\vx_s \by{\tau'} \vx'$ for some prefix $\tau'$  of $\tau$ then $\vx'(h) \in F_s$, i.e., 
the home agent $h$ only participates in atoms of the fragment and is always only ready
to participate in atoms of the fragment.\\

\noindent
{\em Proof of claim.}
The proof follows the iterative construction of $\tau$. We start at marking $\vx_s$,
and we have $\vx_s(h)=s$ because $h$ is a party of $s$ and $\vx_s$ enables $s$.

Whenever $\sigma_e$ or a prefix of $\sigma_e$ occurs, the property is preserved, 
because first, $\next (n,h,r) \in F_s$ holds for all outcomes $(n,r)$ of $\sigma_e$
except the last one (this holds for all parties of $n$); and second,
for the last outcome, which is $(e,r_e)$, $h$ is either not party of $e$ whence the marking
of $h$ does not change, or $\next (e, h, r_e) \in F_s$ by definition of $h$.

Whenever an outcome $(n, p(n))$ occurs, either $h$ is not a party of $n$, and then the marking of 
$h$ does not change, or $h$ is a party of $n$, and $n$ is an atom of $F_s$. 
By construction of $p$ (actually, of $p_i$), the property is preserved, which finishes the 
proof of the claim.\\

Let us now investigate the occurrences of external and internal atoms in $\tau$.
Let $G_E$ be the graph with the external atoms as nodes and an edge
from $n$ to $n'$ if $p_e(n) = n'$. By the definition of $p_e$, the graph $G_E$ is acyclic with
$n_f$ as sink. By the definition of $\tau$, after an external atom $n$ occurs in $\tau$, 
none of its predecessors in $G_E$ can occur in $\tau$. So $\tau$ contains only finitely many occurrences 
of external atoms. 

Since $\tau$ is infinite, it therefore has an infinite suffix $\tau'$ in which only internal
atoms occur. Since $s$ is a synchronizer with a minimal set of parties, every internal agent
participates in infinitely many outcomes of $\tau'$, in particular the home agent $h$. 
By the claim, $\tau'$ contains infinitely many occurrences of atoms of $\F_s$. 

Now let $G_s$ be the graph with the atoms of $F_s$ as nodes, and an edge  
from $n$ to $n'$ if $p_i(n) = n'$. By the definition of $p_i$, every cycle of the graph $G_s$ goes through
the synchronizer $s$. So $\tau'$ contains infinitely many occurrences of $s$.
Whenever $s$ is enabled, $\sigma_e$ is enabled, too, and actually occurs by the definition of $\tau$.
Since $\sigma_e$ ends with the outcome $(e, r_e)$,
 $\tau'$  also contains infinitely many occurrences of $(e, r_e)$. 
Since negotiations have finitely many reachable markings, $\tau'$ contains a loop synchronized by $s$
(by minimality of the synchronizer) and containing $(e,r_e)$. However, by the definition of a fragment
this implies that this loop and thus $(e,r_e)$ belongs to $\F_s$ as well, contradicting that $(e,r_e)$ is an
exit of $\F_s$.
\qed
\end{proof}

In particular, the situation of Figure \ref{fig:exits}(a) cannot occur, and so in 
SDNs the correct picture for the application of the shorcut rule to exits
is the one of Figure \ref{fig:exits}(b): the exit $n'$ has the same agents as the 
synchronizer $s$. Moreover, the new target of $(s, r'_{f})$ equals the 
already existing target of $(n',r')$.
So Lemma \ref{lem:exit} leads to the following bound on the number of targets of $\N_i$:

\begin{lemma}
For every $1 \leq i \leq |N_0|$: $\Ta(\N_{i}) \subseteq \Ta(N_0)$.
\label{lem:targets}
 \end{lemma}
\begin{proof}
It suffices to prove $\Ta(\N_{i+1}) \subseteq \Ta(\N_i)$ for $i < |N_0|$. Let $(n, r)$ be an arbitrary outcome of $\N_{i+1}$.
We show that there exists an outcome $(n',r')$ of $N_i$ such that $(n,r)$ and $(n',r')$ have the same targets. 


 If $(n,r)$ is also an outcome of $N_i$, 
then we are done. So assume this is not the case.
Then $(n, r)$ is generated by a particular application of the shortcut rule during the reduction
process leading from $\N_i$ to $\N_{i+1}$. Let $\N'$ and $\N''$ be the negotiations right
before and after this application of the rule. $\N'$ contains a fragment $\F_s'$
obtained by applying to $\F_s$ the same sequence of rules leading from $\N_i$ to
$\N'$. Similarly, $\N''$ contains a fragment $\F_s''$.

By the definition of the shortcut rule, $\N'$ has an outcome $(n_1, r_1)$ such that
$n_1 $ is an atom of $\F_s'$ and $(n_1, r_1)$ unconditionally enables another atom $n_2$ of $\F_s'$. Moreover,
 $(n, r)$ is the shortcut of $(n_1, r_1)$ and $(n_2, r_2)$, i.e., $(n,r)$ is 
obtained from clause (2) in Definition \ref{def:shortcutrule}. 

We prove the following three claims:

\begin{itemize}
\item[(1)] $(n_1, r_1)$ is an outcome of $\F_s'$, i.e.,  $\next(n_1, a, r_1) \in F_s'$
for every party $a$ of $n_1$. \\
Assume the contrary. Then, since $n_1 \in F_s'$, $(n_1, r_1)$ is an exit
of $\F_s'$, by Lemma \ref{lem:exit} we have $\next(n_1, a, r_1) \notin F_s'$ for every
party $a$ of $n_1$, contradicting that $(n_1, r_1)$ unconditionally enables an atom of $\F_s'$.
\item[(2)] $(n_2, r_2)$ is an exit of $\F_s'$. \\
Assume the contrary, i.e., $(n_2, r_2) \in \F_s'$.
By (1), both $(n_1, r_1)$ and $(n_2, r_2)$ are outcomes of $\F_s'$, 
and so $(n,r)$ is an outcome of $\F_s''$. But then, since $\F_s''$ is completely reduced
by the reductions leading from $\N''$ to $\N_{i+1}$, the outcome $(n, r)$ is removed by 
some rule in the reduction path between $\N''$ and $\N_{i+1}$, contradicting our
assumption that $(n, r)$ is an outcome of $\N_{i+1}$. 
\item[(3)] $(n,r)$ and $(n_2, r_2)$ have the same target.\\
By (2) and Lemma \ref{lem:exit}, $n_2$ has exactly the same parties as the synchronizer $s$. Since
$(n_1,r_1)$ unconditionally enables $n_2$, the same holds for $n_1$. So we have $P_{n_1} = P_{n_2} = P_s$
and $\next(n_1, a, r_1) = n_2$ for every $a \in P_{n_2}$. 
By the definition of the shortcut rule, $\next(n,a,r) =\next(n_2, a, r_2)$ for every $a \in P_{n_2}$, and we are done.
\end{itemize}

To finally prove that $(n,r)$ has the same target as some outcome of $\N_i$ 
we proceed by induction on the number $k$ of times the shortcut rule has
been applied between $\N_i$ and $\N'$. If $k=0$, then $(n_2, r_2)$ is an outcome of 
$\N_i$, and by (3) we are done. If $k > 0$, then either $(n_2, r_2)$ is an outcome of $\N_i$,
and by (3) we are done, or it is produced by a former application of the shortcut rule. In this case,
by induction hypothesis, $(n_2, r_2)$ has  the same target in $\N_i$ and therefore, by
(3), so has $(n,r)$. \qed
\end{proof}

We use this lemma to bound ${\it  Out}(N_i')$.

\begin{lemma}
\label{lem:bounds}
For every $1 \leq i \leq |N_0|$:  $|\Out(N_i')| \in {\cal O}( |N_0|^2 \cdot |\Out(N_0)|)$.
\end{lemma}
\begin{proof}
We first give an upper bound for $|\Out(N_i)|$.
Since the merge rule cannot be applied to $\N_i$, no two outcomes of $\N_i$ have the same source and the same target,
and so $|\Out(N_i)| \leq |N_i| \cdot |\Ta(N_i)|$. By Lemma \ref{lem:targets},
$|\Out(N_i)| \leq |N_0| \cdot |\Out(N_0)|$.

Now we consider $|\Out(N_i')|$. Each outcome of $\Out(\N_i') \setminus \Out(\N_i)$ has some atom of $\F_s$ as source, and is generated by some exit of $\F_s$. So the number of such outcomes is at most the product of the numbers of nodes of $\F_s$ and the number of exits of $\F_s$. Since these numbers are bounded by $|N_i|$ and
$|\Out(\N_i)|$, respectively, we get $ |\Out(N_i')| \leq |\Out(N_i)|+ |N_i|\cdot |\Out(\N_i)|$. The result now follows from 
$|\Out(N_i)| \leq |N_0| \cdot |\Out(N_0)|$ and Lemma \ref{lem:basics}(a).\qed
\end{proof}

Finally, combining Lemma \ref{lem:firstbound} and Lemma \ref{lem:bounds} we get

\begin{theorem}
Let $\N_0$ be an SDN. Then ${\it Appl}(\N_0)  \in   {\cal O} ( \; |N_0|^4 \cdot \Out(\N_0) \;)$.
\end{theorem}

We conjecture that a more detailed complexity analysis can improve this bound
to at least ${\cal O} (|N_0|^3 \cdot \Out(\N_0))$, but this is beyond the scope of this paper.

\section{Conclusions}

We have continued the analysis of negotiations started in \cite{negI}.
We have provided a set of three reduction rules that can
summarize all and only the sound deterministic negotiations. Moreover, the 
number of rule applications is polynomial in teh size of the negotiation.
 
The completeness and polynomiality proofs turned out to be quite involved. 
At the same time, we think they provide interesting insights. In particular,
the completeness proofs shows how in deterministic negotiations
soundness requires to synchronize all agents at least once in every loop. It also
shows that, intuitively, loops must be properly nested. Intuitively, sound
deterministic negotiations are {\em necessarily} well structured, in the sense of
structured programming.

Our rules generalize the rules used to transform 
finite automata into regular expressions by eliminating states \cite{HUM}. 
Indeed, deterministic negotiations can be seen as a class of  communicating 
deterministic automata, and thus our result becomes a generalization 
of Kleene's theorem to a concurrency model. In future work we plan to investigate the 
connection to other concurrent Kleene theorems in the literature like e.g.
\cite{DBLP:conf/icalp/GastinPZ91,DBLP:conf/dlt/GenestMK04}.

\bibliographystyle{is-abbrv} 
\bibliography{references}

\setcounter{lemma}{2}

\section*{Appendix}
\subsection{Proofs of Section \ref{subsec:lassos}}
\begin{lemma}
\label{lem:struct}
\begin{itemize}
\item[(1)] Every cyclic SDN has a loop.
\item[(2)] The set of atoms of a minimal loop generates a strongly connected subgraph of the graph of the considered negotiation.
\end{itemize}
\end{lemma}
\begin{proof}
\begin{itemize}
\item[(1)]
Let $\pi$ be a cycle of the graph of the negotiation ${\cal N}$. Let $n_1$ be an arbitrary atom occurring in $\pi$, and let $n_2$ be its successor in $\pi$. $n_1 \neq n_f$ because $n_f$ has no successor, and hence no cycle contains $n_f$.

By soundness some reachable marking $\vx_1$ enables $n_1$. For at least one agent $a$ and one result $r$, $\next (n_1,a,r)$ contains $n_2$, and 
by determinism  it contains only $n_2$. Let $\vx_1 \by{(n_1,r)} \vx_1'$. 
Again by soundness, there is an occurrence sequence from $\vx_1'$ that leads to the final marking. This sequence has to contain an occurrence of $n_2$ because 
this is the only atom  agent $a$ is ready to engage in. In particular, some prefix of this sequence leads to a marking $\vx_2$ that enables $n_2$. 

Repeating this argument arbitrarily for the nodes $n_1$, $n_2$, $n_3$, \ldots , $n_k = n_1$ of the cycle $\pi$, we conclude that there is an infinite occurrence sequence, containing infinitely many occurrences of atoms of the cycle $\pi$. Since the set of reachable markings is finite, this sequence contains a loop. 
\item[(2)]
For each agent involved in any atom of the loop, consider the sequence of atoms it is involved in. By definition of the graph of the negotiation, this sequence is a path. It is moreover a (not necessarily simple) cycle, because a loop starts and ends with the same marking. So the generated subgraph is covered by cycles. It is moreover strongly connected because, for each connected component, the projection of the outcomes of the loop to those with atoms in this component is a smaller loop, against minimality of the  loop.\qed
\end{itemize}
\end{proof}

\subsection{Proofs of Section \ref{subsec:synchronizers}}
\begin{lemma}
\label{lem:struct2}
Every minimal loop of a SDN is synchronized. 
\end{lemma}
\begin{proof}

Let $\sigma$ be a minimal loop, enabled at a reachable marking $\vx$.
Define $N_\sigma$ as the set of atoms that occur in $\sigma$
and $A_\sigma$ as the set of agents involved in atoms of $N_\sigma$.
Since $\N$ is sound, there is an occurrence sequence $\sigma_f$ enabled by $\vx$ that ends with the final atom $n_f$. 

Now choose  an  agent $\hat{a}$ of $A_\sigma$ such that $\vx (\hat{a}) \in N_\sigma$.
In $\sigma_f$,  eventually $\hat{a}$ is involved in $n_f$, and it is first involved in an atom of $N_\sigma$.

Using $\sigma_f$, we
construct a path $\pi$ of the graph of $\N$ as follows:
We begin this path with the last atom $n \in N_\sigma$ that appears in $\sigma_f$ and involves agent $\hat{a}$. We call this atom $n_\pi$. Then we
repeatedly choose the last atom in $\sigma_f$ that involves $\hat{a}$ and moreover is a successor of the last vertex of the path constructed so far. 
By construction, this path has no cycles (i.e., all vertices are distinct), starts with an atom of $N_\sigma$ and has not further atoms of $N_\sigma$, ends with $n_f$, and only contains atoms involving $\hat{a}$.

Since $\vx$ enables the loop $\sigma$ and since $n_\pi \in N_\sigma$, after some prefix of $\sigma$ a marking $\vx_\pi$ is reached which enables $n_\pi$. The loop
$\sigma$ continues with some outcome $(n_\pi,r_1)$, where $r_1$ is one possible result of $n_\pi$. 

By construction of the path $\pi$, 
there is an alternative result $r_2$ of $n_\pi$ such that $\next (n_\pi,\hat{a},r_2)$ is the second atom of the path $\pi$, and this atom does not belong to $N_\sigma$. 
Let $\vx_\pi'$ be the marking reached after the occurrence of $(n_\pi,r_2)$ at
 $\vx_\pi$.

From $\vx_\pi'$,  we iteratively construct an occurrence sequence as follows:
\begin{itemize}
\item[(1)] if an atom $n$ of $N_\sigma$ is enabled and thus some $(n,r)$ occurs in $\sigma$, we continue with $(n,r)$,
\item[(2)] otherwise, if an atom $n$ of the path $\pi$ is enabled, we let this atom occur with an outcome $r$ such that $\next (n,\hat{a},r)$ is the successor atom w.r.t.\  the path $\pi$,
\item[(3)] otherwise we add a minimal occurrence sequence that either leads to the final marking or enables an atom of $\sigma$ or an atom of $\pi$, so that after this sequence one of the previous rules can be applied. Such an occurrence sequence exists because $\N$ is sound and hence the final marking can be reached.
\end{itemize}

First observe that, {\jde in the constructed sequence}, agent $\hat{a}$ will always be ready to engage only in an atom of the path $\pi$. So its token is moved along  $\pi$. Conversely, all atoms of $\pi$ involve $\hat{a}$. Therefore only finitely many atoms of $\pi$ occur in the sequence. This limits the total number of occurrences of type (2).

Agent $\hat{a}$ is no more ready to engage in any atom of $N_\sigma$ during the sequence. So at least $n_\pi$ cannot occur any more in the sequence because $\hat{a}$ is a party of $n_\pi$. By minimality of the loop $\sigma$, there is no loop with a set of atoms in $N_\sigma \setminus \{n_\pi\}$. Since the set of reachable markings is finite, there cannot be an infinite sequence of outcomes with atoms of $N_\sigma$ (type (1)) without occurrences of other atoms. 

By determinism, each agent ready to engage in an atom of $N_\sigma$ 
can only engage in this atom. So the set of these agents is only changed by occurrences of type (1). By construction, no agent ever leaves the loop after the occurrence of $(n_\pi,r_2)$, i.e.\ every agent of this set remains in this set by an occurrence of type (1). Therefore, the set of agents ready to engage in an atom of $N_\sigma$ never decreases.

For each sequence of type (3) we have three possibilities. 
\begin{itemize}
\item[(a)]
It  ends with the final marking. 
\item[(b)]
It ends with a marking that enables an atom of $\pi$ (type (2)), which then occurs next. However these atoms  can  occur only finitely often in the constructed sequence, as already mentioned.
\item[(c)]
It ends  with a marking that enables an atom of $\sigma$ (type (1)) which then occurs next. In that case the last outcome of this sequence necessarily involves an agent of $A_\sigma$, which after this occurrence is ready to engage in an atom of $N_\sigma$. So it increases the number of agents ready to engage in an atom of $N_\sigma$. Since this number never decreases, this option can also happen only finitely often.
\end{itemize}

Hence, eventually only option (a) is possible, and so the sequence  will reach the final marking.
Since the final atom involves all agents, no agent was able to remain in the loop.
In other words: all agents of $A_\sigma$ left the loop when $(n_\pi, r_2)$ has occurred.
As a consequence, all these agents are parties of $n_\pi$, and $n_\pi$ therefore is a synchronizer of the loop $\sigma$.
\qed
\end{proof}

\subsection{Proofs of Section \ref{subsec:fragments}}
\begin{proposition}\label{prop:occseq}
\begin{itemize}
\item[(1)] If $(\rho,\sigma)$ is a lasso of $\N$ and $\sigma$ is
synchronized by $s$, then $(\epsilon, \sigma)$ is a lasso of $\N_s$.
\item[(2)] If $\sigma$ is an occurrence sequence of $\N_s$ leading to a marking $\vx$, 
then there is an occurrence sequence $\rho$ of $\N$ such that $\rho \sigma$ 
is an occurrence sequence of $\N$ leading to a marking $\vx'$ such that $\vx$ is the projection of $\vx' $ onto the parties of $s$ (remember that all parties of atoms of $\N_s$ are parties of $s$ because $s$ is a synchronizer.
\item[(3)]
If $\N$ is a SDN, then  $\N_s$ is a SDN.
\end{itemize}
\end{proposition}
\begin{proof}
Let $n'$ be an atom of $\N_n$. By definition, there is a lasso $(\rho,\sigma)$ of
$\N$, synchronized by $n$, such that $n'$ appears in $\sigma$. By Lemma \ref{lem:occseq}(1) 
$\sigma$ is an occurrence sequence of $\N_n$, and so $n'$ can occur in $\N_n$.

Assume now that $\sigma$ is an occurrence sequence of $\N_n$ that is not a large 
step of $\N_n$ and cannot be extended to a large step of $\N_n$. 
W.l.o.g. we can assume that $\sigma$ does not enable any atom of $\N_n$. 
Let $\vx$ be the marking reached by $\sigma$. By Lemma \ref{lem:occseq}(2) there is 
an occurrence sequence $\rho \sigma$ of $\N$. 
Let $\vx$ be the marking reached by this sequence. By Lemma \ref{lem:occseq}(2),
$\vx'$ is the projection of $\vx$ onto the set $P$ of parties of $n$.

By soundness there is a maximal
occurrence sequence $\rho\rho'$ such that $\rho'$ contains only atoms with parties
in $\agents\setminus P_n$. Clearly $\rho\rho'\sigma$ is an occurrence sequence of $\N$.
Let $\vx''$ be the marking reached by $\rho\rho'\sigma$. Clearly, we still have that
$\vx'$ is the projection of $\vx''$ onto $P$. We claim that $\vx''$ is a deadlock, 
contradicting the soundness of $\N$. To prove the claim, assume that $\vx''$ 
enables some atom $n'$ with set of parties $P'$. If $P' \subseteq P_t$, then 
$n'$ is also enabled at $\vx'$, contradicting that $\sigma$ does not enable any atom of $\N_n$. If $P' \subseteq \agents\setminus P$, then $\rho\rho'$ enables $n'$,
contradicting the maximality of $\rho\rho'$. Finally, if $P' \cap P \neq \emptyset P' \cap (\agents\setminus P)$, then there is an agent $a \in P$ such that $\vx''(a) \notin N_n$. But by definition of $\sigma$ we have $\vx'(a) \in N_n$, contradicting that $\vx'$ is the
projection onto $P$ of $\vx''$. \qed
\end{proof}

\begin{proposition}\label{prop:subsub}
Let $\N_{n_1}$ be the projection of $\N$ on ${\cal L}_{n_1}$, and let
$\N_{n_1n_2}$ is the projection of $\N_{n_1}$ on the set ${\cal L}_{n_2}$ of loops of
$\N_{n_1}$. Then $\N_{n_1n_2}$ and $\N_{n_2}$ are isomorphic.
\end{proposition}
\begin{proof}
It suffices to show that every loop of $\N_{n_1n_2}$ is contained
in some loop of $\N_{n_2}$ and viceversa.

Let $\sigma$ be a loop of $\N_{n_1n_2}$. Then, by definition, $\sigma$ is also
a loop of $\N_{n_2}$. Now, let $\sigma$ be a loop of $\N_{n_2}$
Since $n_2$ is an atom of $\N_{n_1}$, some loop of $\N_{n_1}$ 
has the form $\tau_1 (n_2,r) \tau_2$. But then 
$\tau_1 \sigma (n_2,r) \tau_2$ is also a loop of $\N_{n_1}$, and so 
$\sigma$ is a loop of $\N_{n_1n_2}$. \qed
\end{proof}

\begin{lemma}
A cyclic SDN contains an atom $n$ such that $\N_n$ is an acyclic SDN. 
\label{lem:soundsubneg}
\end{lemma}
\begin{proof}
Let $\N$ be a cyclic SDN. By Lemma \ref{lem:struct} $\N$ has a loop and hence also a minimal loop. By
Lemma \ref{lem:struct2} this loop has a synchronizer $n$,  and so  $\N_n$ is nonempty.
Choose $n$ so that $\N_n$ is nonempty, but its number of atoms is minimal. We claim that
$\N_n$ is acyclic. Assume the contrary. By Proposition \ref{prop:occseq}, $\N_n$ is
a SDN. By Lemmas \ref{lem:struct} and \ref{lem:struct2}, exactly as above, $\N_n$ contains an atom $n'$ such that $\N_{nn'}$ is nonempty. Clearly, 
$\N_{nn'}$ contains fewer atoms than $\N_n$ and by Proposition \ref{prop:subsub} it is isomorphic to $\N_{n'}$. 
This contradicts the minimality of $\N_n$. \qed
\end{proof}

\begin{lemma}
A cyclic SDN $\N$ contains an atom $n$ such that $\N_n$ is an acyclic SDN. 
\label{lem:soundsubneg}
\end{lemma}
\begin{proof}
Let $\N$ be a cyclic SDN. By Lemma \ref{lem:struct}(1) $\N$ has a loop and hence also a minimal loop. By
Lemma \ref{lem:struct2} this loop has at least one synchronizer. 
We choose a synchronizer $s$ such that $\N_s$ (which is nonempty for each synchronizer) is minimal.

First we argue that $\N_s$ is sound. 
Otherwise either some atom can never become enabled, which is impossible because $\N_s$ 
is built from loops synchronized by $s$, or a deadlock marking is reached after some occurrence 
sequence $\sigma$. At this marking, we necessarily have that
for each atom $n \in N_s$ some party of $n$ is ready to engage in some
different $n' \in N_s$, and by determinism only there 
(remember that all agents of $\N_s$ are parties of $s$).
In $\N$, starting with some reachable marking that enables $s$, the sequence $\sigma$ 
can occur, too. It leads to a marking with the same property: for each atom 
$n \in N_s$ some party of $n$ is ready to engage in a different atom of $\N_s$. 
But them, no matter which atoms outside $\N_s$ occur, the atoms of $\N_s$ 
can never become enabled again. In particular, since all agents of $\N_s$ are parties of $n_f$, 
we have that after $\sigma$ the final atom $n_f$ can never occur, 
contradicting soundness of $\N$.

Next we claim that $\N_s$ is acyclic. Assume the contrary. 
Since $\N_n$ is a SDN, we can argue as above and find 
a loop and an atom $n'$ such that $\N_{nn'}$ is nonempty. Clearly, 
$\N_{nn'}$ contains fewer atoms than $\N_n$.
Now each loop of $\N_s$ is also a loop of $\N$,
and $\N_{nn'}$ is equal to $\N_{n'}$. 
This contradicts the minimality of $\N_n$. \qed
\end{proof}

\subsection{Proofs of Section \ref{subsec:targets}}

We prove Lemma \ref{lem:exit}. Recall that, intuitively, the lemma states that
the occurrence of an exit $(e, r_e)$ of $\F_s$ forces all agents of $P_s$ to leave the fragment $\F_s$. In other words: all agents of $P_s$ are parties of $n$, and 
the occurrence of $(e,r_e)$ does not lead any agent back to an atom of $\F_s$.

\begin{lemma}
\label{lem:exit}
Let $\F_s$ be a fragment of a SDN $\N$, and let $(e, r_e)$ be an exit of $\F_s$. 
Then $P_e=P_s$ (i.e., $e$ has the same parties as $s$), and $\next(e,a,r_e) \notin F_s$ for every $a\in P_e$. 
\end{lemma}
\begin{proof}
We proceed indirectly and assume that either $P_e \subset P_s$ 
($P_e \subseteq P_s$ by the definition of fragment) or 
$\next(e,a,r_e) \in F_s$ for some $a\in P_e$. Then
at least one agent $h \in P_s$ satisfies 
either $h \notin P_e$ or $\next(e,h,r_e) \in F_s$. We call 
$h$ a {\em home agent} (intuitively, an agent that does not 
leave ``home'', i.e., $\F_s$,  by the occurrence of the exit). 
We show that the existence of $h$ leads to a contradiction.

We partition the set of $\agents$ of agents into internal agents, the agents
of $P_s$, and external agents, the agents of $\agents \setminus P_s$. 
We also partition the set of atoms: an atom is internal if it has only 
internal parties, otherwise it is external. Clearly all atoms of $F_s$ 
are internal, but there can also be internal atoms outside $F_s$. 
If $P_s$ contains all agents of the negotiation, then all agents are internal,
and so are all atoms (also the final atom $n_f$). 
Otherwise at least $n_f$ has an external party and is hence an external atom.

Next we define a function $p \colon N \to Out (\N)$ that assigns to each atom one of 
its outcomes (the {\em preferred} outcome). $p$ is defined for internal and external 
atoms separately, i.e., it is the union of functions $p_i$ 
assigning outcomes to internal atoms, and $p_e$ assigning outcomes to external atoms.

If there are external atoms, and hence $n_f$ is external, $p_e$ is defined as follows. 
First we set $p_e(n_f)$ to an arbitrary outcome of $n_f$. 
Then we proceed iteratively: If some external atom $n$ has an outcome $r$ and an external agent 
$a$ such that $p_e(\next(n,a,r))$ is defined, then set $p_e(n):=r$ (if there are several possibilities, 
we choose one of them arbitrarily).
At the end of the procedure $p_e$ is defined for every external atom, 
because each external atom $n$ has an external agent, say $a$, and, 
since $a$ participates in $n_f$,
the graph of $\N$ has a path of atoms, all of them with $a$ as party, 
leading from $n$
to $n_f$. 

Now we define $p_i$ for internal atoms. 
For the internal atoms $n$ not in $\F_s$ we define $p_i(n)$
arbitrarily. For the internal atoms $n \in F_s$ such that $\next(n,a,r)=s$ 
for some agent $a$ we set $p_i(n)=r$. For the rest of the internal atoms of $F_s$
we proceed iteratively. If $n \in F_s$ has an outcome $r$ and an agent $a$ 
(necessarily internal) such that $p_i(\next(n,a,r))$ is defined, 
then we set $p_i(n,a):=r$ (if there are several 
possibilities, we choose one of them). By Lemma \ref{lem:struct2}(2), the graph
of $\F_s$ is strongly connected, and so eventually $p_i$ is defined for 
all atoms of $\F_s$.

Let $\sigma$ be an arbitrary occurrence sequence leading to a marking $\vx_s$
that enables $s$ (remember that $\N$ is sound). By the definition
of $\F_s$, the marking $\vx_s$ enables an occurrence sequence $\sigma_e$ 
that starts with an occurrence of $s$, contains only atoms of
$F_s$, and ends with an occurrence of $(e, r_e)$, the considered 
exit of $\F_s$.

We now define a  {\em maximal} occurrence sequence $\tau$ enabled at $\vx_s$.
We start with $\tau:=\epsilon$ and while $\tau$ enables some atom
proceed iteratively as follows:
\begin{itemize}
\item If $\tau$ enables $\sigma_e$, then $\tau:=\tau\sigma_e$, i.e., we 
extend the current sequence with $\sigma_e$.
\item Otherwise, choose 
any enabled atom $n$, and set $\tau:= \tau (n,p(n))$, i.e., 
we extend the current sequence with  $(n,p(n)$.
\end{itemize} 

We first show that $\tau$ is infinite, i.e., that we never exit the while loop.
By soundness, there is always an 
enabled atom as long as the final marking is not reached, i.e., as long 
as at least one agent is ready to engage in an atom. 
So it suffices to show that this is the case. We prove that the home agent 
$h$ is ready to engage in an atom after the occurrence of an arbitrary finite prefix of $\tau$.
This result follows from the following claim.\\

\noindent
{\em Claim.}
If $\vx_s \by{\tau'} \vx'$ for some prefix $\tau'$  of $\tau$ then $\vx'(h) \in F_s$, i.e., 
the home agent $h$ only participates in atoms of the fragment and is always only ready
to participate in atoms of the fragment.\\

\noindent
{\em Proof of claim.}
The proof follows the iterative construction of $\tau$. We start at marking $\vx_s$,
and we have $\vx_s(h)=s$ because $h$ is a party of $s$ and $\vx_s$ enables $s$.

Whenever $\sigma_e$ or a prefix of $\sigma_e$ occurs, the property is preserved, 
because first, $\next (n,h,r) \in F_s$ holds for all outcomes $(n,r)$ of $\sigma_e$
except the last one (this holds for all parties of $n$); and second,
for the last outcome, which is $(e,r_e)$, $h$ is either not party of $e$ whence the marking
of $h$ does not change, or $\next (e, h, r_e) \in F_s$ by definition of $h$.

Whenever an outcome $(n, p(n))$ occurs, either $h$ is not a party of $n$, and then the marking of 
$h$ does not change, or $h$ is a party of $n$, and $n$ is an atom of $F_s$. 
By construction of $p$ (actually, of $p_i$), the property is preserved, which finishes the 
proof of the claim.\\

Let us now investigate the occurrences of external and internal atoms in $\tau$.
Let $G_E$ be the graph with the external atoms as nodes and an edge
from $n$ to $n'$ if $p_e(n) = n'$. By the definition of $p_e$, the graph $G_E$ is acyclic with
$n_f$ as sink. By the definition of $\tau$, after an external atom $n$ occurs in $\tau$, 
none of its predecessors in $G_E$ can occur in $\tau$. So $\tau$ contains only finitely many occurrences 
of external atoms. 

Since $\tau$ is infinite, it therefore has an infinite suffix $\tau'$ in which only internal
atoms occur. Since $s$ is a synchronizer with a minimal set of parties, every internal agent
participates in infinitely many outcomes of $\tau'$, in particular the home agent $h$. 
By the claim, $\tau'$ contains infinitely many occurrences of atoms of $\F_s$. 

Now let $G_s$ be the graph with the atoms of $F_s$ as nodes, and an edge  
from $n$ to $n'$ if $p_i(n) = n'$. By the definition of $p_i$, every cycle of the graph $G_s$ goes through
the synchronizer $s$. So $\tau'$ contains infinitely many occurrences of $s$.
Whenever $s$ is enabled, $\sigma_e$ is enabled, too, and actually occurs by the definition of $\tau$.
Since $\sigma_e$ ends with the outcome $(e, r_e)$,
 $\tau'$  also contains infinitely many occurrences of $(e, r_e)$. 
Since negotiations have finitely many reachable markings, $\tau'$ contains a loop synchronized by $s$
(by minimality of the synchronizer) and containing $(e,r_e)$. However, by the definition of a fragment
this implies that this loop and thus $(e,r_e)$ belongs to $\F_s$ as well, contradicting that $(e,r_e)$ is an
exit of $\F_s$.
\qed
\end{proof}

\begin{lemma}
\label{lem:targets}
For every $1 \leq i \leq |N_0|$: $\Ta(\N_{i}) \subseteq \Ta(N_0)$.
 \end{lemma}
\begin{proof}
It suffices to prove $\Ta(\N_{i+1}) \subseteq \Ta(\N_i)$ for $i < |N_0|$. Let $(n, r)$ be an arbitrary outcome of $\N_{i+1}$.
We show that there exists an outcome $(n',r')$ of $N_i$ such that $(n,r)$ and $(n',r')$ have the same targets. 


 If $(n,r)$ is also an outcome of $N_i$, 
then we are done. So assume this is not the case.
Then $(n, r)$ is generated by a particular application of the shortcut rule during the reduction
process leading from $\N_i$ to $\N_{i+1}$. Let $\N'$ and $\N''$ be the negotiations right
before and after this application of the rule. $\N'$ contains a fragment $\F_s'$
obtained by applying to $\F_s$ the same sequence of rules leading from $\N_i$ to
$\N'$. Similarly, $\N''$ contains a fragment $\F_s''$.

By the definition of the shortcut rule, $\N'$ has an outcome $(n_1, r_1)$ such that
$n_1 $ is an atom of $\F_s'$ and $(n_1, r_1)$ unconditionally enables another atom $n_2$ of $\F_s'$. Moreover,
 $(n, r)$ is the shortcut of $(n_1, r_1)$ and $(n_2, r_2)$, i.e., $(n,r)$ is 
obtained from clause (2) in Definition \ref{def:shortcutrule}. 

We prove the following three claims:

\begin{itemize}
\item[(1)] $(n_1, r_1)$ is an outcome of $\F_s'$, i.e.,  $\next(n_1, a, r_1) \in F_s'$
for every party $a$ of $n_1$. \\
Assume the contrary. Then, since $n_1 \in F_s'$, $(n_1, r_1)$ is an exit
of $\F_s'$, by Lemma \ref{lem:exit} we have $\next(n_1, a, r_1) \notin F_s'$ for every
party $a$ of $n_1$, contradicting that $(n_1, r_1)$ unconditionally enables an atom of $\F_s'$.
\item[(2)] $(n_2, r_2)$ is an exit of $\F_s'$. \\
Assume the contrary, i.e., $(n_2, r_2) \in \F_s'$.
By (1), both $(n_1, r_1)$ and $(n_2, r_2)$ are outcomes of $\F_s'$, 
and so $(n,r)$ is an outcome of $\F_s''$. But then, since $\F_s''$ is completely reduced
by the reductions leading from $\N''$ to $\N_{i+1}$, the outcome $(n, r)$ is removed by 
some rule in the reduction path between $\N''$ and $\N_{i+1}$, contradicting our
assumption that $(n, r)$ is an outcome of $\N_{i+1}$. 
\item[(3)] $(n,r)$ and $(n_2, r_2)$ have the same target.\\
By (2) and Lemma \ref{lem:exit}, $n_2$ has exactly the same parties as the synchronizer $s$. Since
$(n_1,r_1)$ unconditionally enables $n_2$, the same holds for $n_1$. So we have $P_{n_1} = P_{n_2} = P_s$
and $\next(n_1, a, r_1) = n_2$ for every $a \in P_{n_2}$. 
By the definition of the shortcut rule, $\next(n,a,r) =\next(n_2, a, r_2)$ for every $a \in P_{n_2}$, and we are done.
\end{itemize}

To finally prove that $(n,r)$ has the same target as some outcome of $\N_i$ 
we proceed by induction on the number $k$ of times the shortcut rule has
been applied between $\N_i$ and $\N'$. If $k=0$, then $(n_2, r_2)$ is an outcome of 
$\N_i$, and by (3) we are done. If $k > 0$, then either $(n_2, r_2)$ is an outcome of $\N_i$,
and by (3) we are done, or it is produced by a former application of the shortcut rule. In this case,
by induction hypothesis, $(n_2, r_2)$ has  the same target in $\N_i$ and therefore, by
(3), so has $(n,r)$. \qed
\end{proof}
\end{document}

%% file: distnegotII.pdf_t
\begin{picture}(0,0)%
\includegraphics{distnegotII.pdf}%
\end{picture}%
%
%
\setlength{\unitlength}{4144sp}%
\begingroup\makeatletter\ifx\SetFigFont\undefined%
\gdef\SetFigFont#1#2#3#4#5{%
  \reset@font\fontsize{#1}{#2pt}%
  \fontfamily{#3}\fontseries{#4}\fontshape{#5}%
  \selectfont}%
\fi\endgroup%
\begin{picture}(4669,5919)(7909,-6874)
\put(8236,-4756){\makebox(0,0)[rb]{\smash{{\SetFigFont{20}{24.0}{\rmdefault}{\mddefault}{\updefault}{\color[rgb]{0,0,0}\texttt{y,n,am}}%
}}}}
\put(8256,-2073){\makebox(0,0)[rb]{\smash{{\SetFigFont{20}{24.0}{\rmdefault}{\mddefault}{\updefault}{\color[rgb]{0,0,0}\texttt{st}}%
}}}}
\put(12061,-2086){\makebox(0,0)[lb]{\smash{{\SetFigFont{20}{24.0}{\rmdefault}{\mddefault}{\updefault}{\color[rgb]{0,0,0}\texttt{st}}%
}}}}
\put(9726,-2103){\makebox(0,0)[lb]{\smash{{\SetFigFont{20}{24.0}{\rmdefault}{\mddefault}{\updefault}{\color[rgb]{0,0,0}\texttt{st}}%
}}}}
\put(7999,-2892){\makebox(0,0)[rb]{\smash{{\SetFigFont{25}{30.0}{\rmdefault}{\mddefault}{\updefault}{\color[rgb]{0,0,0}$n_1$}%
}}}}
\put(9921,-3468){\makebox(0,0)[lb]{\smash{{\SetFigFont{20}{24.0}{\rmdefault}{\mddefault}{\updefault}{\color[rgb]{0,0,0}\texttt{am}}%
}}}}
\put(9556,-4426){\makebox(0,0)[lb]{\smash{{\SetFigFont{20}{24.0}{\rmdefault}{\mddefault}{\updefault}{\color[rgb]{0,0,0}\texttt{y,n}}%
}}}}
\put(10696,-5643){\makebox(0,0)[rb]{\smash{{\SetFigFont{20}{24.0}{\rmdefault}{\mddefault}{\updefault}{\color[rgb]{0,0,0}\texttt{y,n}}%
}}}}
\put(11839,-5656){\makebox(0,0)[rb]{\smash{{\SetFigFont{20}{24.0}{\rmdefault}{\mddefault}{\updefault}{\color[rgb]{0,0,0}\texttt{y,n}}%
}}}}
\put(10939,-4919){\makebox(0,0)[rb]{\smash{{\SetFigFont{25}{30.0}{\rmdefault}{\mddefault}{\updefault}{\color[rgb]{0,0,0}$n_2$}%
}}}}
\put(7924,-1334){\makebox(0,0)[rb]{\smash{{\SetFigFont{25}{30.0}{\rmdefault}{\mddefault}{\updefault}{\color[rgb]{0,0,0}$n_0$}%
}}}}
\put(7949,-6747){\makebox(0,0)[rb]{\smash{{\SetFigFont{25}{30.0}{\rmdefault}{\mddefault}{\updefault}{\color[rgb]{0,0,0}$n_f$}%
}}}}
\end{picture}%

%% file: detfdm.pdf_t
\begin{picture}(0,0)%
\includegraphics{detfdm.pdf}%
\end{picture}%
%
%
\setlength{\unitlength}{4144sp}%
\begingroup\makeatletter\ifx\SetFigFont\undefined%
\gdef\SetFigFont#1#2#3#4#5{%
  \reset@font\fontsize{#1}{#2pt}%
  \fontfamily{#3}\fontseries{#4}\fontshape{#5}%
  \selectfont}%
\fi\endgroup%
\begin{picture}(6240,5018)(6987,-5739)
\put(8025,-5612){\makebox(0,0)[rb]{\smash{{\SetFigFont{25}{30.0}{\rmdefault}{\mddefault}{\updefault}{\color[rgb]{0,0,0}$n_f$}%
}}}}
\put(10077,-1684){\makebox(0,0)[rb]{\smash{{\SetFigFont{20}{24.0}{\rmdefault}{\mddefault}{\updefault}{\color[rgb]{0,0,0}\texttt{y}}%
}}}}
\put(8412,-1684){\makebox(0,0)[lb]{\smash{{\SetFigFont{20}{24.0}{\rmdefault}{\mddefault}{\updefault}{\color[rgb]{0,0,0}\texttt{y}}%
}}}}
\put(11837,-1686){\makebox(0,0)[rb]{\smash{{\SetFigFont{20}{24.0}{\rmdefault}{\mddefault}{\updefault}{\color[rgb]{0,0,0}\texttt{y}}%
}}}}
\put(8025,-2417){\makebox(0,0)[rb]{\smash{{\SetFigFont{25}{30.0}{\rmdefault}{\mddefault}{\updefault}{\color[rgb]{0,0,0}$n_1$}%
}}}}
\put(7965,-1052){\makebox(0,0)[rb]{\smash{{\SetFigFont{25}{30.0}{\rmdefault}{\mddefault}{\updefault}{\color[rgb]{0,0,0}$n_0$}%
}}}}
\put(8040,-4127){\makebox(0,0)[rb]{\smash{{\SetFigFont{25}{30.0}{\rmdefault}{\mddefault}{\updefault}{\color[rgb]{0,0,0}$n_2$}%
}}}}
\put(7002,-3199){\makebox(0,0)[rb]{\smash{{\SetFigFont{20}{24.0}{\rmdefault}{\mddefault}{\updefault}{\color[rgb]{0,0,0}\texttt{n}}%
}}}}
\put(11307,-3199){\makebox(0,0)[lb]{\smash{{\SetFigFont{20}{24.0}{\rmdefault}{\mddefault}{\updefault}{\color[rgb]{0,0,0}\texttt{n}}%
}}}}
\put(11822,-4696){\makebox(0,0)[rb]{\smash{{\SetFigFont{20}{24.0}{\rmdefault}{\mddefault}{\updefault}{\color[rgb]{0,0,0}\texttt{y}}%
}}}}
\put(8367,-4694){\makebox(0,0)[lb]{\smash{{\SetFigFont{20}{24.0}{\rmdefault}{\mddefault}{\updefault}{\color[rgb]{0,0,0}\texttt{y}}%
}}}}
\put(10032,-4679){\makebox(0,0)[rb]{\smash{{\SetFigFont{20}{24.0}{\rmdefault}{\mddefault}{\updefault}{\color[rgb]{0,0,0}\texttt{y}}%
}}}}
\put(8397,-3214){\makebox(0,0)[lb]{\smash{{\SetFigFont{20}{24.0}{\rmdefault}{\mddefault}{\updefault}{\color[rgb]{0,0,0}\texttt{tm}}%
}}}}
\put(10022,-3216){\makebox(0,0)[rb]{\smash{{\SetFigFont{20}{24.0}{\rmdefault}{\mddefault}{\updefault}{\color[rgb]{0,0,0}\texttt{tm}}%
}}}}
\put(13212,-3229){\makebox(0,0)[lb]{\smash{{\SetFigFont{20}{24.0}{\rmdefault}{\mddefault}{\updefault}{\color[rgb]{0,0,0}\texttt{n}}%
}}}}
\put(7887,-3379){\makebox(0,0)[rb]{\smash{{\SetFigFont{20}{24.0}{\rmdefault}{\mddefault}{\updefault}{\color[rgb]{0,0,0}\texttt{r}}%
}}}}
\put(10572,-3379){\makebox(0,0)[lb]{\smash{{\SetFigFont{20}{24.0}{\rmdefault}{\mddefault}{\updefault}{\color[rgb]{0,0,0}\texttt{r}}%
}}}}
\put(12477,-3574){\makebox(0,0)[lb]{\smash{{\SetFigFont{20}{24.0}{\rmdefault}{\mddefault}{\updefault}{\color[rgb]{0,0,0}\texttt{r}}%
}}}}
\end{picture}%

%% file: subneg3.pdf_t
\begin{picture}(0,0)%
\includegraphics{subneg3.pdf}%
\end{picture}%
%
%
\setlength{\unitlength}{4144sp}%
\begingroup\makeatletter\ifx\SetFigFontNFSS\undefined%
\gdef\SetFigFontNFSS#1#2#3#4#5{%
  \reset@font\fontsize{#1}{#2pt}%
  \fontfamily{#3}\fontseries{#4}\fontshape{#5}%
  \selectfont}%
\fi\endgroup%
\begin{picture}(16067,15628)(192,-10001)
\put(3739,-5847){\makebox(0,0)[rb]{\smash{{\SetFigFontNFSS{20}{24.0}{\rmdefault}{\mddefault}{\updefault}{\color[rgb]{0,0,1}$n_4$}%
}}}}
\put(3724,-4482){\makebox(0,0)[rb]{\smash{{\SetFigFontNFSS{20}{24.0}{\rmdefault}{\mddefault}{\updefault}{\color[rgb]{0,0,1}$n_2$}%
}}}}
\put(4200,-5171){\makebox(0,0)[lb]{\smash{{\SetFigFontNFSS{20}{24.0}{\rmdefault}{\mddefault}{\updefault}{\color[rgb]{0,0,1}$a$}%
}}}}
\put(4219,-6537){\makebox(0,0)[lb]{\smash{{\SetFigFontNFSS{20}{24.0}{\rmdefault}{\mddefault}{\updefault}{\color[rgb]{0,0,1}$b$}%
}}}}
\put(2566,3959){\makebox(0,0)[rb]{\smash{{\SetFigFontNFSS{20}{24.0}{\rmdefault}{\mddefault}{\updefault}{\color[rgb]{0,0,0}$a$}%
}}}}
\put(4036,-1441){\makebox(0,0)[rb]{\smash{{\SetFigFontNFSS{20}{24.0}{\rmdefault}{\mddefault}{\updefault}{\color[rgb]{0,0,0}$a$}%
}}}}
\put(2697,-1425){\makebox(0,0)[rb]{\smash{{\SetFigFontNFSS{20}{24.0}{\rmdefault}{\mddefault}{\updefault}{\color[rgb]{0,0,0}$a$}%
}}}}
\put(5326,3959){\makebox(0,0)[rb]{\smash{{\SetFigFontNFSS{20}{24.0}{\rmdefault}{\mddefault}{\updefault}{\color[rgb]{0,0,0}$a$}%
}}}}
\put(3946,3974){\makebox(0,0)[rb]{\smash{{\SetFigFontNFSS{20}{24.0}{\rmdefault}{\mddefault}{\updefault}{\color[rgb]{0,0,0}$a$}%
}}}}
\put(3421,3524){\makebox(0,0)[b]{\smash{{\SetFigFontNFSS{20}{24.0}{\rmdefault}{\mddefault}{\updefault}{\color[rgb]{0,0,0}$n_1$}%
}}}}
\put(3451,2789){\makebox(0,0)[rb]{\smash{{\SetFigFontNFSS{20}{24.0}{\rmdefault}{\mddefault}{\updefault}{\color[rgb]{0,0,0}$b$}%
}}}}
\put(2176,2804){\makebox(0,0)[rb]{\smash{{\SetFigFontNFSS{20}{24.0}{\rmdefault}{\mddefault}{\updefault}{\color[rgb]{0,0,0}$b$}%
}}}}
\put(3357,-570){\makebox(0,0)[b]{\smash{{\SetFigFontNFSS{20}{24.0}{\rmdefault}{\mddefault}{\updefault}{\color[rgb]{0,0,0}$n_5$}%
}}}}
\put(3976, 14){\makebox(0,0)[rb]{\smash{{\SetFigFontNFSS{20}{24.0}{\rmdefault}{\mddefault}{\updefault}{\color[rgb]{0,0,0}$a$}%
}}}}
\put(4621,-271){\makebox(0,0)[lb]{\smash{{\SetFigFontNFSS{20}{24.0}{\rmdefault}{\mddefault}{\updefault}{\color[rgb]{0,0,0}$b$}%
}}}}
\put(631,-286){\makebox(0,0)[rb]{\smash{{\SetFigFontNFSS{20}{24.0}{\rmdefault}{\mddefault}{\updefault}{\color[rgb]{0,0,0}$b$}%
}}}}
\put(4006,2549){\makebox(0,0)[rb]{\smash{{\SetFigFontNFSS{20}{24.0}{\rmdefault}{\mddefault}{\updefault}{\color[rgb]{0,0,0}$a$}%
}}}}
\put(2465,674){\makebox(0,0)[rb]{\smash{{\SetFigFontNFSS{20}{24.0}{\rmdefault}{\mddefault}{\updefault}{\color[rgb]{0,0,0}$a$}%
}}}}
\put(1190,659){\makebox(0,0)[rb]{\smash{{\SetFigFontNFSS{20}{24.0}{\rmdefault}{\mddefault}{\updefault}{\color[rgb]{0,0,0}$a$}%
}}}}
\put(2667,2550){\makebox(0,0)[rb]{\smash{{\SetFigFontNFSS{20}{24.0}{\rmdefault}{\mddefault}{\updefault}{\color[rgb]{0,0,0}$a$}%
}}}}
\put(1550,1409){\makebox(0,0)[b]{\smash{{\SetFigFontNFSS{20}{24.0}{\rmdefault}{\mddefault}{\updefault}{\color[rgb]{0,0,0}$n_3$}%
}}}}
\put(3691,539){\makebox(0,0)[rb]{\smash{{\SetFigFontNFSS{20}{24.0}{\rmdefault}{\mddefault}{\updefault}{\color[rgb]{0,0,1}$n_4$}%
}}}}
\put(4152,1215){\makebox(0,0)[lb]{\smash{{\SetFigFontNFSS{20}{24.0}{\rmdefault}{\mddefault}{\updefault}{\color[rgb]{0,0,1}$a$}%
}}}}
\put(3616,1214){\makebox(0,0)[rb]{\smash{{\SetFigFontNFSS{20}{24.0}{\rmdefault}{\mddefault}{\updefault}{\color[rgb]{0,0,1}$b$}%
}}}}
\put(3676,1904){\makebox(0,0)[rb]{\smash{{\SetFigFontNFSS{20}{24.0}{\rmdefault}{\mddefault}{\updefault}{\color[rgb]{0,0,1}$n_2$}%
}}}}
\put(3496,4874){\makebox(0,0)[b]{\smash{{\SetFigFontNFSS{20}{24.0}{\rmdefault}{\mddefault}{\updefault}{\color[rgb]{0,0,0}$n_0$}%
}}}}
\put(4801,-2581){\makebox(0,0)[b]{\smash{{\SetFigFontNFSS{20}{24.0}{\rmdefault}{\mddefault}{\updefault}{\color[rgb]{0,0,0}$n_f$}%
}}}}
\put(4039,5384){\makebox(0,0)[b]{\smash{{\SetFigFontNFSS{20}{24.0}{\rmdefault}{\mddefault}{\updefault}{\color[rgb]{0,0,0}(a)}%
}}}}
\put(9904,5369){\makebox(0,0)[b]{\smash{{\SetFigFontNFSS{20}{24.0}{\rmdefault}{\mddefault}{\updefault}{\color[rgb]{0,0,0}(c)}%
}}}}
\put(14644,5384){\makebox(0,0)[b]{\smash{{\SetFigFontNFSS{20}{24.0}{\rmdefault}{\mddefault}{\updefault}{\color[rgb]{0,0,0}(e)}%
}}}}
\put(4084,-3676){\makebox(0,0)[b]{\smash{{\SetFigFontNFSS{20}{24.0}{\rmdefault}{\mddefault}{\updefault}{\color[rgb]{0,0,0}(b)}%
}}}}
\put(9319,-3721){\makebox(0,0)[b]{\smash{{\SetFigFontNFSS{20}{24.0}{\rmdefault}{\mddefault}{\updefault}{\color[rgb]{0,0,0}(d)}%
}}}}
\end{picture}%

%% file: cyclicnoloops.pdf_t
\begin{picture}(0,0)%
\includegraphics{cyclicnoloops.pdf}%
\end{picture}%
%
%
\setlength{\unitlength}{4144sp}%
\begingroup\makeatletter\ifx\SetFigFont\undefined%
\gdef\SetFigFont#1#2#3#4#5{%
  \reset@font\fontsize{#1}{#2pt}%
  \fontfamily{#3}\fontseries{#4}\fontshape{#5}%
  \selectfont}%
\fi\endgroup%
\begin{picture}(9346,4551)(6016,-8887)
\put(10801,-4694){\makebox(0,0)[rb]{\smash{{\SetFigFont{29}{34.8}{\rmdefault}{\mddefault}{\updefault}{\color[rgb]{0,0,0}$n_0$}%
}}}}
\put(10786,-8744){\makebox(0,0)[rb]{\smash{{\SetFigFont{29}{34.8}{\rmdefault}{\mddefault}{\updefault}{\color[rgb]{0,0,0}$n_f$}%
}}}}
\put(14986,-5866){\makebox(0,0)[lb]{\smash{{\SetFigFont{29}{34.8}{\rmdefault}{\mddefault}{\updefault}{\color[rgb]{0,0,0}$n_1$}%
}}}}
\put(6031,-8746){\makebox(0,0)[rb]{\smash{{\SetFigFont{29}{34.8}{\rmdefault}{\mddefault}{\updefault}{\color[rgb]{0,0,0}$n_f$}%
}}}}
\put(6031,-4696){\makebox(0,0)[rb]{\smash{{\SetFigFont{29}{34.8}{\rmdefault}{\mddefault}{\updefault}{\color[rgb]{0,0,0}$n_0$}%
}}}}
\put(8686,-6046){\makebox(0,0)[lb]{\smash{{\SetFigFont{29}{34.8}{\rmdefault}{\mddefault}{\updefault}{\color[rgb]{0,0,0}$n_1$}%
}}}}
\put(6031,-7396){\makebox(0,0)[rb]{\smash{{\SetFigFont{29}{34.8}{\rmdefault}{\mddefault}{\updefault}{\color[rgb]{0,0,0}$n_2$}%
}}}}
\put(10801,-7216){\makebox(0,0)[rb]{\smash{{\SetFigFont{29}{34.8}{\rmdefault}{\mddefault}{\updefault}{\color[rgb]{0,0,0}$n_2$}%
}}}}
\end{picture}%

%% file: cyclicfragment.pdf_t
\begin{picture}(0,0)%
\includegraphics{cyclicfragment.pdf}%
\end{picture}%
%
%
\setlength{\unitlength}{4144sp}%
\begingroup\makeatletter\ifx\SetFigFontNFSS\undefined%
\gdef\SetFigFontNFSS#1#2#3#4#5{%
  \reset@font\fontsize{#1}{#2pt}%
  \fontfamily{#3}\fontseries{#4}\fontshape{#5}%
  \selectfont}%
\fi\endgroup%
\begin{picture}(13335,6603)(6241,-2311)
\put(9470,3524){\makebox(0,0)[b]{\smash{{\SetFigFontNFSS{20}{24.0}{\rmdefault}{\mddefault}{\updefault}{\color[rgb]{0,0,0}$n_1$}%
}}}}
\put(9500,2789){\makebox(0,0)[rb]{\smash{{\SetFigFontNFSS{20}{24.0}{\rmdefault}{\mddefault}{\updefault}{\color[rgb]{0,0,0}$b$}%
}}}}
\put(8225,2804){\makebox(0,0)[rb]{\smash{{\SetFigFontNFSS{20}{24.0}{\rmdefault}{\mddefault}{\updefault}{\color[rgb]{0,0,0}$b$}%
}}}}
\put(9406,-570){\makebox(0,0)[b]{\smash{{\SetFigFontNFSS{20}{24.0}{\rmdefault}{\mddefault}{\updefault}{\color[rgb]{0,0,0}$n_5$}%
}}}}
\put(10025, 14){\makebox(0,0)[rb]{\smash{{\SetFigFontNFSS{20}{24.0}{\rmdefault}{\mddefault}{\updefault}{\color[rgb]{0,0,0}$a$}%
}}}}
\put(10670,-271){\makebox(0,0)[lb]{\smash{{\SetFigFontNFSS{20}{24.0}{\rmdefault}{\mddefault}{\updefault}{\color[rgb]{0,0,0}$b$}%
}}}}
\put(6680,-286){\makebox(0,0)[rb]{\smash{{\SetFigFontNFSS{20}{24.0}{\rmdefault}{\mddefault}{\updefault}{\color[rgb]{0,0,0}$b$}%
}}}}
\put(9740,539){\makebox(0,0)[rb]{\smash{{\SetFigFontNFSS{20}{24.0}{\rmdefault}{\mddefault}{\updefault}{\color[rgb]{0,0,0}$n_4$}%
}}}}
\put(10055,2549){\makebox(0,0)[rb]{\smash{{\SetFigFontNFSS{20}{24.0}{\rmdefault}{\mddefault}{\updefault}{\color[rgb]{0,0,0}$a$}%
}}}}
\put(10201,1215){\makebox(0,0)[lb]{\smash{{\SetFigFontNFSS{20}{24.0}{\rmdefault}{\mddefault}{\updefault}{\color[rgb]{0,0,0}$a$}%
}}}}
\put(9665,1214){\makebox(0,0)[rb]{\smash{{\SetFigFontNFSS{20}{24.0}{\rmdefault}{\mddefault}{\updefault}{\color[rgb]{0,0,0}$b$}%
}}}}
\put(8514,674){\makebox(0,0)[rb]{\smash{{\SetFigFontNFSS{20}{24.0}{\rmdefault}{\mddefault}{\updefault}{\color[rgb]{0,0,0}$a$}%
}}}}
\put(7239,659){\makebox(0,0)[rb]{\smash{{\SetFigFontNFSS{20}{24.0}{\rmdefault}{\mddefault}{\updefault}{\color[rgb]{0,0,0}$a$}%
}}}}
\put(9725,1904){\makebox(0,0)[rb]{\smash{{\SetFigFontNFSS{20}{24.0}{\rmdefault}{\mddefault}{\updefault}{\color[rgb]{0,0,0}$n_2$}%
}}}}
\put(8716,2550){\makebox(0,0)[rb]{\smash{{\SetFigFontNFSS{20}{24.0}{\rmdefault}{\mddefault}{\updefault}{\color[rgb]{0,0,0}$a$}%
}}}}
\put(7599,1409){\makebox(0,0)[b]{\smash{{\SetFigFontNFSS{20}{24.0}{\rmdefault}{\mddefault}{\updefault}{\color[rgb]{0,0,0}$n_3$}%
}}}}
\put(13981,3539){\makebox(0,0)[b]{\smash{{\SetFigFontNFSS{20}{24.0}{\rmdefault}{\mddefault}{\updefault}{\color[rgb]{0,0,0}$n_1$}%
}}}}
\put(14011,2804){\makebox(0,0)[rb]{\smash{{\SetFigFontNFSS{20}{24.0}{\rmdefault}{\mddefault}{\updefault}{\color[rgb]{0,0,0}$b$}%
}}}}
\put(12736,2819){\makebox(0,0)[rb]{\smash{{\SetFigFontNFSS{20}{24.0}{\rmdefault}{\mddefault}{\updefault}{\color[rgb]{0,0,0}$b$}%
}}}}
\put(13917,-555){\makebox(0,0)[b]{\smash{{\SetFigFontNFSS{20}{24.0}{\rmdefault}{\mddefault}{\updefault}{\color[rgb]{0,0,0}$n_5$}%
}}}}
\put(14536, 29){\makebox(0,0)[rb]{\smash{{\SetFigFontNFSS{20}{24.0}{\rmdefault}{\mddefault}{\updefault}{\color[rgb]{0,0,0}$a$}%
}}}}
\put(14251,554){\makebox(0,0)[rb]{\smash{{\SetFigFontNFSS{20}{24.0}{\rmdefault}{\mddefault}{\updefault}{\color[rgb]{0,0,0}$n_4$}%
}}}}
\put(14566,2564){\makebox(0,0)[rb]{\smash{{\SetFigFontNFSS{20}{24.0}{\rmdefault}{\mddefault}{\updefault}{\color[rgb]{0,0,0}$a$}%
}}}}
\put(14712,1230){\makebox(0,0)[lb]{\smash{{\SetFigFontNFSS{20}{24.0}{\rmdefault}{\mddefault}{\updefault}{\color[rgb]{0,0,0}$a$}%
}}}}
\put(14176,1229){\makebox(0,0)[rb]{\smash{{\SetFigFontNFSS{20}{24.0}{\rmdefault}{\mddefault}{\updefault}{\color[rgb]{0,0,0}$b$}%
}}}}
\put(13025,689){\makebox(0,0)[rb]{\smash{{\SetFigFontNFSS{20}{24.0}{\rmdefault}{\mddefault}{\updefault}{\color[rgb]{0,0,0}$a$}%
}}}}
\put(11750,674){\makebox(0,0)[rb]{\smash{{\SetFigFontNFSS{20}{24.0}{\rmdefault}{\mddefault}{\updefault}{\color[rgb]{0,0,0}$a$}%
}}}}
\put(14236,1919){\makebox(0,0)[rb]{\smash{{\SetFigFontNFSS{20}{24.0}{\rmdefault}{\mddefault}{\updefault}{\color[rgb]{0,0,0}$n_2$}%
}}}}
\put(13227,2565){\makebox(0,0)[rb]{\smash{{\SetFigFontNFSS{20}{24.0}{\rmdefault}{\mddefault}{\updefault}{\color[rgb]{0,0,0}$a$}%
}}}}
\put(12110,1424){\makebox(0,0)[b]{\smash{{\SetFigFontNFSS{20}{24.0}{\rmdefault}{\mddefault}{\updefault}{\color[rgb]{0,0,0}$n_3$}%
}}}}
\put(13216,-1426){\makebox(0,0)[rb]{\smash{{\SetFigFontNFSS{20}{24.0}{\rmdefault}{\mddefault}{\updefault}{\color[rgb]{0,0,0}$b$}%
}}}}
\put(14716,-1411){\makebox(0,0)[lb]{\smash{{\SetFigFontNFSS{20}{24.0}{\rmdefault}{\mddefault}{\updefault}{\color[rgb]{0,0,0}$b$}%
}}}}
\put(16540,528){\makebox(0,0)[rb]{\smash{{\SetFigFontNFSS{20}{24.0}{\rmdefault}{\mddefault}{\updefault}{\color[rgb]{0,0,0}$n_4$}%
}}}}
\put(17001,1204){\makebox(0,0)[lb]{\smash{{\SetFigFontNFSS{20}{24.0}{\rmdefault}{\mddefault}{\updefault}{\color[rgb]{0,0,0}$a$}%
}}}}
\put(16465,1203){\makebox(0,0)[rb]{\smash{{\SetFigFontNFSS{20}{24.0}{\rmdefault}{\mddefault}{\updefault}{\color[rgb]{0,0,0}$b$}%
}}}}
\put(16525,1893){\makebox(0,0)[rb]{\smash{{\SetFigFontNFSS{20}{24.0}{\rmdefault}{\mddefault}{\updefault}{\color[rgb]{0,0,0}$n_2$}%
}}}}
\put(18751,528){\makebox(0,0)[rb]{\smash{{\SetFigFontNFSS{20}{24.0}{\rmdefault}{\mddefault}{\updefault}{\color[rgb]{0,0,0}$n_4$}%
}}}}
\put(19212,1204){\makebox(0,0)[lb]{\smash{{\SetFigFontNFSS{20}{24.0}{\rmdefault}{\mddefault}{\updefault}{\color[rgb]{0,0,0}$a$}%
}}}}
\put(18736,1893){\makebox(0,0)[rb]{\smash{{\SetFigFontNFSS{20}{24.0}{\rmdefault}{\mddefault}{\updefault}{\color[rgb]{0,0,0}$n_2$}%
}}}}
\put(19231,-162){\makebox(0,0)[lb]{\smash{{\SetFigFontNFSS{20}{24.0}{\rmdefault}{\mddefault}{\updefault}{\color[rgb]{0,0,0}$b$}%
}}}}
\put(9406,4049){\makebox(0,0)[b]{\smash{{\SetFigFontNFSS{20}{24.0}{\rmdefault}{\mddefault}{\updefault}{\color[rgb]{0,0,0}(a)}%
}}}}
\put(13861,4049){\makebox(0,0)[b]{\smash{{\SetFigFontNFSS{20}{24.0}{\rmdefault}{\mddefault}{\updefault}{\color[rgb]{0,0,0}(b)}%
}}}}
\put(16891,2684){\makebox(0,0)[b]{\smash{{\SetFigFontNFSS{20}{24.0}{\rmdefault}{\mddefault}{\updefault}{\color[rgb]{0,0,0}(c)}%
}}}}
\put(19111,2684){\makebox(0,0)[b]{\smash{{\SetFigFontNFSS{20}{24.0}{\rmdefault}{\mddefault}{\updefault}{\color[rgb]{0,0,0}(d)}%
}}}}
\end{picture}%

%% file: oneagent.pdf_t
\begin{picture}(0,0)%
\includegraphics{oneagent.pdf}%
\end{picture}%
%
%
\setlength{\unitlength}{4144sp}%
\begingroup\makeatletter\ifx\SetFigFont\undefined%
\gdef\SetFigFont#1#2#3#4#5{%
  \reset@font\fontsize{#1}{#2pt}%
  \fontfamily{#3}\fontseries{#4}\fontshape{#5}%
  \selectfont}%
\fi\endgroup%
\begin{picture}(18620,12241)(391,-12497)
\put(16656,-9164){\makebox(0,0)[rb]{\smash{{\SetFigFont{25}{30.0}{\rmdefault}{\mddefault}{\updefault}{\color[rgb]{0,0,1}$n_1$}%
}}}}
\put(2791,-6044){\makebox(0,0)[lb]{\smash{{\SetFigFont{25}{30.0}{\rmdefault}{\mddefault}{\updefault}{\color[rgb]{0,0,0}$n_1$}%
}}}}
\put(406,-1964){\makebox(0,0)[rb]{\smash{{\SetFigFont{25}{30.0}{\rmdefault}{\mddefault}{\updefault}{\color[rgb]{0,0,1}$n_1$}%
}}}}
\put(436,-4694){\makebox(0,0)[rb]{\smash{{\SetFigFont{25}{30.0}{\rmdefault}{\mddefault}{\updefault}{\color[rgb]{0,0,1}$n_4$}%
}}}}
\put(4081,-1979){\makebox(0,0)[lb]{\smash{{\SetFigFont{25}{30.0}{\rmdefault}{\mddefault}{\updefault}{\color[rgb]{0,0,0}$n_2$}%
}}}}
\put(2731,-614){\makebox(0,0)[lb]{\smash{{\SetFigFont{25}{30.0}{\rmdefault}{\mddefault}{\updefault}{\color[rgb]{0,0,0}$n_0$}%
}}}}
\put(4111,-4694){\makebox(0,0)[lb]{\smash{{\SetFigFont{25}{30.0}{\rmdefault}{\mddefault}{\updefault}{\color[rgb]{0,0,0}$n_6$}%
}}}}
\put(2671,-4679){\makebox(0,0)[lb]{\smash{{\SetFigFont{25}{30.0}{\rmdefault}{\mddefault}{\updefault}{\color[rgb]{0,0,0}$n_5$}%
}}}}
\put(2701,-3299){\makebox(0,0)[lb]{\smash{{\SetFigFont{25}{30.0}{\rmdefault}{\mddefault}{\updefault}{\color[rgb]{0,0,1}$n_3$}%
}}}}
\put(1396,-1169){\makebox(0,0)[rb]{\smash{{\SetFigFont{25}{30.0}{\rmdefault}{\mddefault}{\updefault}{\color[rgb]{0,0,0}$a$}%
}}}}
\put(2161,-5249){\makebox(0,0)[rb]{\smash{{\SetFigFont{25}{30.0}{\rmdefault}{\mddefault}{\updefault}{\color[rgb]{0,0,0}$a$}%
}}}}
\put(3061,-5429){\makebox(0,0)[lb]{\smash{{\SetFigFont{25}{30.0}{\rmdefault}{\mddefault}{\updefault}{\color[rgb]{0,0,0}$a$}%
}}}}
\put(826,-3329){\makebox(0,0)[rb]{\smash{{\SetFigFont{25}{30.0}{\rmdefault}{\mddefault}{\updefault}{\color[rgb]{0,0,1}$b$}%
}}}}
\put(2956,-1169){\makebox(0,0)[lb]{\smash{{\SetFigFont{25}{30.0}{\rmdefault}{\mddefault}{\updefault}{\color[rgb]{0,0,0}$b$}%
}}}}
\put(1561,-2459){\makebox(0,0)[lb]{\smash{{\SetFigFont{25}{30.0}{\rmdefault}{\mddefault}{\updefault}{\color[rgb]{0,0,1}$a$}%
}}}}
\put(3226,-2489){\makebox(0,0)[lb]{\smash{{\SetFigFont{25}{30.0}{\rmdefault}{\mddefault}{\updefault}{\color[rgb]{0,0,0}$a$}%
}}}}
\put(1411,-5414){\makebox(0,0)[rb]{\smash{{\SetFigFont{25}{30.0}{\rmdefault}{\mddefault}{\updefault}{\color[rgb]{0,0,0}$a$}%
}}}}
\put(2161,-4004){\makebox(0,0)[rb]{\smash{{\SetFigFont{25}{30.0}{\rmdefault}{\mddefault}{\updefault}{\color[rgb]{0,0,0}$b$}%
}}}}
\put(2836,-3839){\makebox(0,0)[lb]{\smash{{\SetFigFont{25}{30.0}{\rmdefault}{\mddefault}{\updefault}{\color[rgb]{0,0,0}$c$}%
}}}}
\put(1561,-3824){\makebox(0,0)[rb]{\smash{{\SetFigFont{25}{30.0}{\rmdefault}{\mddefault}{\updefault}{\color[rgb]{0,0,1}$a$}%
}}}}
\put(1786,-7169){\makebox(0,0)[rb]{\smash{{\SetFigFont{25}{30.0}{\rmdefault}{\mddefault}{\updefault}{\color[rgb]{0,0,1}$n_1$}%
}}}}
\put(2761,-8564){\makebox(0,0)[lb]{\smash{{\SetFigFont{25}{30.0}{\rmdefault}{\mddefault}{\updefault}{\color[rgb]{0,0,1}$n_3$}%
}}}}
\put(1831,-9869){\makebox(0,0)[rb]{\smash{{\SetFigFont{25}{30.0}{\rmdefault}{\mddefault}{\updefault}{\color[rgb]{0,0,1}$n_4$}%
}}}}
\put(1831,-11219){\makebox(0,0)[rb]{\smash{{\SetFigFont{25}{30.0}{\rmdefault}{\mddefault}{\updefault}{\color[rgb]{0,0,1}$n_1'$}%
}}}}
\put(2326,-7799){\makebox(0,0)[lb]{\smash{{\SetFigFont{25}{30.0}{\rmdefault}{\mddefault}{\updefault}{\color[rgb]{0,0,1}$a$}%
}}}}
\put(2191,-9119){\makebox(0,0)[rb]{\smash{{\SetFigFont{25}{30.0}{\rmdefault}{\mddefault}{\updefault}{\color[rgb]{0,0,1}$a$}%
}}}}
\put(2206,-10484){\makebox(0,0)[rb]{\smash{{\SetFigFont{25}{30.0}{\rmdefault}{\mddefault}{\updefault}{\color[rgb]{0,0,1}$a$}%
}}}}
\put(7732,-6059){\makebox(0,0)[lb]{\smash{{\SetFigFont{25}{30.0}{\rmdefault}{\mddefault}{\updefault}{\color[rgb]{0,0,0}$n_1$}%
}}}}
\put(5347,-1979){\makebox(0,0)[rb]{\smash{{\SetFigFont{25}{30.0}{\rmdefault}{\mddefault}{\updefault}{\color[rgb]{0,0,1}$n_1$}%
}}}}
\put(5377,-4709){\makebox(0,0)[rb]{\smash{{\SetFigFont{25}{30.0}{\rmdefault}{\mddefault}{\updefault}{\color[rgb]{0,0,1}$n_4$}%
}}}}
\put(7672,-629){\makebox(0,0)[lb]{\smash{{\SetFigFont{25}{30.0}{\rmdefault}{\mddefault}{\updefault}{\color[rgb]{0,0,0}$n_0$}%
}}}}
\put(7612,-4694){\makebox(0,0)[lb]{\smash{{\SetFigFont{25}{30.0}{\rmdefault}{\mddefault}{\updefault}{\color[rgb]{0,0,0}$n_5$}%
}}}}
\put(7642,-3314){\makebox(0,0)[lb]{\smash{{\SetFigFont{25}{30.0}{\rmdefault}{\mddefault}{\updefault}{\color[rgb]{0,0,1}$n_3$}%
}}}}
\put(6337,-1184){\makebox(0,0)[rb]{\smash{{\SetFigFont{25}{30.0}{\rmdefault}{\mddefault}{\updefault}{\color[rgb]{0,0,0}$a$}%
}}}}
\put(7102,-5264){\makebox(0,0)[rb]{\smash{{\SetFigFont{25}{30.0}{\rmdefault}{\mddefault}{\updefault}{\color[rgb]{0,0,0}$a$}%
}}}}
\put(8002,-5444){\makebox(0,0)[lb]{\smash{{\SetFigFont{25}{30.0}{\rmdefault}{\mddefault}{\updefault}{\color[rgb]{0,0,0}$a$}%
}}}}
\put(7897,-1184){\makebox(0,0)[lb]{\smash{{\SetFigFont{25}{30.0}{\rmdefault}{\mddefault}{\updefault}{\color[rgb]{0,0,0}$b$}%
}}}}
\put(8167,-2504){\makebox(0,0)[lb]{\smash{{\SetFigFont{25}{30.0}{\rmdefault}{\mddefault}{\updefault}{\color[rgb]{0,0,0}$a$}%
}}}}
\put(6352,-5429){\makebox(0,0)[rb]{\smash{{\SetFigFont{25}{30.0}{\rmdefault}{\mddefault}{\updefault}{\color[rgb]{0,0,0}$a$}%
}}}}
\put(7102,-4019){\makebox(0,0)[rb]{\smash{{\SetFigFont{25}{30.0}{\rmdefault}{\mddefault}{\updefault}{\color[rgb]{0,0,0}$b$}%
}}}}
\put(7777,-3854){\makebox(0,0)[lb]{\smash{{\SetFigFont{25}{30.0}{\rmdefault}{\mddefault}{\updefault}{\color[rgb]{0,0,0}$c$}%
}}}}
\put(6502,-3839){\makebox(0,0)[rb]{\smash{{\SetFigFont{25}{30.0}{\rmdefault}{\mddefault}{\updefault}{\color[rgb]{0,0,1}$a$}%
}}}}
\put(6727,-7184){\makebox(0,0)[rb]{\smash{{\SetFigFont{25}{30.0}{\rmdefault}{\mddefault}{\updefault}{\color[rgb]{0,0,1}$n_1$}%
}}}}
\put(6772,-9884){\makebox(0,0)[rb]{\smash{{\SetFigFont{25}{30.0}{\rmdefault}{\mddefault}{\updefault}{\color[rgb]{0,0,1}$n_4$}%
}}}}
\put(6772,-11234){\makebox(0,0)[rb]{\smash{{\SetFigFont{25}{30.0}{\rmdefault}{\mddefault}{\updefault}{\color[rgb]{0,0,1}$n_1'$}%
}}}}
\put(7147,-10499){\makebox(0,0)[rb]{\smash{{\SetFigFont{25}{30.0}{\rmdefault}{\mddefault}{\updefault}{\color[rgb]{0,0,1}$a$}%
}}}}
\put(7252,-8519){\makebox(0,0)[lb]{\smash{{\SetFigFont{25}{30.0}{\rmdefault}{\mddefault}{\updefault}{\color[rgb]{1,0,0}$a_1$}%
}}}}
\put(5647,-3329){\makebox(0,0)[rb]{\smash{{\SetFigFont{25}{30.0}{\rmdefault}{\mddefault}{\updefault}{\color[rgb]{0,0,1}$b$}%
}}}}
\put(5947,-3344){\makebox(0,0)[lb]{\smash{{\SetFigFont{25}{30.0}{\rmdefault}{\mddefault}{\updefault}{\color[rgb]{1,0,0}$a_1$}%
}}}}
\put(6382,-2789){\makebox(0,0)[lb]{\smash{{\SetFigFont{25}{30.0}{\rmdefault}{\mddefault}{\updefault}{\color[rgb]{1,0,0}$a_2$}%
}}}}
\put(6682,-2084){\makebox(0,0)[lb]{\smash{{\SetFigFont{25}{30.0}{\rmdefault}{\mddefault}{\updefault}{\color[rgb]{1,0,0}$a_3$}%
}}}}
\put(9022,-1994){\makebox(0,0)[lb]{\smash{{\SetFigFont{25}{30.0}{\rmdefault}{\mddefault}{\updefault}{\color[rgb]{0,0,0}$n_2$}%
}}}}
\put(9052,-4709){\makebox(0,0)[lb]{\smash{{\SetFigFont{25}{30.0}{\rmdefault}{\mddefault}{\updefault}{\color[rgb]{0,0,0}$n_6$}%
}}}}
\put(12711,-6074){\makebox(0,0)[lb]{\smash{{\SetFigFont{25}{30.0}{\rmdefault}{\mddefault}{\updefault}{\color[rgb]{0,0,0}$n_1$}%
}}}}
\put(10356,-4724){\makebox(0,0)[rb]{\smash{{\SetFigFont{25}{30.0}{\rmdefault}{\mddefault}{\updefault}{\color[rgb]{0,0,1}$n_4$}%
}}}}
\put(12651,-644){\makebox(0,0)[lb]{\smash{{\SetFigFont{25}{30.0}{\rmdefault}{\mddefault}{\updefault}{\color[rgb]{0,0,0}$n_0$}%
}}}}
\put(14031,-4724){\makebox(0,0)[lb]{\smash{{\SetFigFont{25}{30.0}{\rmdefault}{\mddefault}{\updefault}{\color[rgb]{0,0,0}$n_6$}%
}}}}
\put(12591,-4709){\makebox(0,0)[lb]{\smash{{\SetFigFont{25}{30.0}{\rmdefault}{\mddefault}{\updefault}{\color[rgb]{0,0,0}$n_5$}%
}}}}
\put(12621,-3329){\makebox(0,0)[lb]{\smash{{\SetFigFont{25}{30.0}{\rmdefault}{\mddefault}{\updefault}{\color[rgb]{0,0,1}$n_3$}%
}}}}
\put(11316,-1199){\makebox(0,0)[rb]{\smash{{\SetFigFont{25}{30.0}{\rmdefault}{\mddefault}{\updefault}{\color[rgb]{0,0,0}$a$}%
}}}}
\put(12081,-5279){\makebox(0,0)[rb]{\smash{{\SetFigFont{25}{30.0}{\rmdefault}{\mddefault}{\updefault}{\color[rgb]{0,0,0}$a$}%
}}}}
\put(12981,-5459){\makebox(0,0)[lb]{\smash{{\SetFigFont{25}{30.0}{\rmdefault}{\mddefault}{\updefault}{\color[rgb]{0,0,0}$a$}%
}}}}
\put(12876,-1199){\makebox(0,0)[lb]{\smash{{\SetFigFont{25}{30.0}{\rmdefault}{\mddefault}{\updefault}{\color[rgb]{0,0,0}$b$}%
}}}}
\put(13146,-2519){\makebox(0,0)[lb]{\smash{{\SetFigFont{25}{30.0}{\rmdefault}{\mddefault}{\updefault}{\color[rgb]{0,0,0}$a$}%
}}}}
\put(11331,-5444){\makebox(0,0)[rb]{\smash{{\SetFigFont{25}{30.0}{\rmdefault}{\mddefault}{\updefault}{\color[rgb]{0,0,0}$a$}%
}}}}
\put(12081,-4034){\makebox(0,0)[rb]{\smash{{\SetFigFont{25}{30.0}{\rmdefault}{\mddefault}{\updefault}{\color[rgb]{0,0,0}$b$}%
}}}}
\put(12756,-3869){\makebox(0,0)[lb]{\smash{{\SetFigFont{25}{30.0}{\rmdefault}{\mddefault}{\updefault}{\color[rgb]{0,0,0}$c$}%
}}}}
\put(11706,-7199){\makebox(0,0)[rb]{\smash{{\SetFigFont{25}{30.0}{\rmdefault}{\mddefault}{\updefault}{\color[rgb]{0,0,1}$n_1$}%
}}}}
\put(11751,-11249){\makebox(0,0)[rb]{\smash{{\SetFigFont{25}{30.0}{\rmdefault}{\mddefault}{\updefault}{\color[rgb]{0,0,1}$n_1'$}%
}}}}
\put(10626,-3344){\makebox(0,0)[rb]{\smash{{\SetFigFont{25}{30.0}{\rmdefault}{\mddefault}{\updefault}{\color[rgb]{0,0,1}$b$}%
}}}}
\put(11361,-2804){\makebox(0,0)[lb]{\smash{{\SetFigFont{25}{30.0}{\rmdefault}{\mddefault}{\updefault}{\color[rgb]{1,0,0}$a_2$}%
}}}}
\put(12246,-9209){\makebox(0,0)[lb]{\smash{{\SetFigFont{25}{30.0}{\rmdefault}{\mddefault}{\updefault}{\color[rgb]{1,0,0}$a_5$}%
}}}}
\put(11256,-4109){\makebox(0,0)[rb]{\smash{{\SetFigFont{25}{30.0}{\rmdefault}{\mddefault}{\updefault}{\color[rgb]{0,0,1}$a$}%
}}}}
\put(11226,-3509){\makebox(0,0)[rb]{\smash{{\SetFigFont{25}{30.0}{\rmdefault}{\mddefault}{\updefault}{\color[rgb]{1,0,0}$a_4$}%
}}}}
\put(11196,-1799){\makebox(0,0)[lb]{\smash{{\SetFigFont{25}{30.0}{\rmdefault}{\mddefault}{\updefault}{\color[rgb]{0,0,1}$n_1$}%
}}}}
\put(9996,-1454){\makebox(0,0)[rb]{\smash{{\SetFigFont{25}{30.0}{\rmdefault}{\mddefault}{\updefault}{\color[rgb]{1,0,0}$a_5$}%
}}}}
\put(11631,-2189){\makebox(0,0)[lb]{\smash{{\SetFigFont{25}{30.0}{\rmdefault}{\mddefault}{\updefault}{\color[rgb]{1,0,0}$a_3$}%
}}}}
\put(13987,-1994){\makebox(0,0)[lb]{\smash{{\SetFigFont{25}{30.0}{\rmdefault}{\mddefault}{\updefault}{\color[rgb]{0,0,0}$n_2$}%
}}}}
\put(14017,-4709){\makebox(0,0)[lb]{\smash{{\SetFigFont{25}{30.0}{\rmdefault}{\mddefault}{\updefault}{\color[rgb]{0,0,0}$n_6$}%
}}}}
\put(17676,-6074){\makebox(0,0)[lb]{\smash{{\SetFigFont{25}{30.0}{\rmdefault}{\mddefault}{\updefault}{\color[rgb]{0,0,0}$n_1$}%
}}}}
\put(15321,-4724){\makebox(0,0)[rb]{\smash{{\SetFigFont{25}{30.0}{\rmdefault}{\mddefault}{\updefault}{\color[rgb]{0,0,1}$n_4$}%
}}}}
\put(18966,-2009){\makebox(0,0)[lb]{\smash{{\SetFigFont{25}{30.0}{\rmdefault}{\mddefault}{\updefault}{\color[rgb]{0,0,0}$n_2$}%
}}}}
\put(17616,-644){\makebox(0,0)[lb]{\smash{{\SetFigFont{25}{30.0}{\rmdefault}{\mddefault}{\updefault}{\color[rgb]{0,0,0}$n_0$}%
}}}}
\put(18996,-4724){\makebox(0,0)[lb]{\smash{{\SetFigFont{25}{30.0}{\rmdefault}{\mddefault}{\updefault}{\color[rgb]{0,0,0}$n_6$}%
}}}}
\put(17556,-4709){\makebox(0,0)[lb]{\smash{{\SetFigFont{25}{30.0}{\rmdefault}{\mddefault}{\updefault}{\color[rgb]{0,0,0}$n_5$}%
}}}}
\put(17586,-3329){\makebox(0,0)[lb]{\smash{{\SetFigFont{25}{30.0}{\rmdefault}{\mddefault}{\updefault}{\color[rgb]{0,0,1}$n_3$}%
}}}}
\put(16281,-1199){\makebox(0,0)[rb]{\smash{{\SetFigFont{25}{30.0}{\rmdefault}{\mddefault}{\updefault}{\color[rgb]{0,0,0}$a$}%
}}}}
\put(17046,-5279){\makebox(0,0)[rb]{\smash{{\SetFigFont{25}{30.0}{\rmdefault}{\mddefault}{\updefault}{\color[rgb]{0,0,0}$a$}%
}}}}
\put(17946,-5459){\makebox(0,0)[lb]{\smash{{\SetFigFont{25}{30.0}{\rmdefault}{\mddefault}{\updefault}{\color[rgb]{0,0,0}$a$}%
}}}}
\put(17841,-1199){\makebox(0,0)[lb]{\smash{{\SetFigFont{25}{30.0}{\rmdefault}{\mddefault}{\updefault}{\color[rgb]{0,0,0}$b$}%
}}}}
\put(18111,-2519){\makebox(0,0)[lb]{\smash{{\SetFigFont{25}{30.0}{\rmdefault}{\mddefault}{\updefault}{\color[rgb]{0,0,0}$a$}%
}}}}
\put(16296,-5444){\makebox(0,0)[rb]{\smash{{\SetFigFont{25}{30.0}{\rmdefault}{\mddefault}{\updefault}{\color[rgb]{0,0,0}$a$}%
}}}}
\put(17046,-4034){\makebox(0,0)[rb]{\smash{{\SetFigFont{25}{30.0}{\rmdefault}{\mddefault}{\updefault}{\color[rgb]{0,0,0}$b$}%
}}}}
\put(17721,-3869){\makebox(0,0)[lb]{\smash{{\SetFigFont{25}{30.0}{\rmdefault}{\mddefault}{\updefault}{\color[rgb]{0,0,0}$c$}%
}}}}
\put(15591,-3344){\makebox(0,0)[rb]{\smash{{\SetFigFont{25}{30.0}{\rmdefault}{\mddefault}{\updefault}{\color[rgb]{0,0,1}$b$}%
}}}}
\put(16326,-2804){\makebox(0,0)[lb]{\smash{{\SetFigFont{25}{30.0}{\rmdefault}{\mddefault}{\updefault}{\color[rgb]{1,0,0}$a_2'$}%
}}}}
\put(16221,-4109){\makebox(0,0)[rb]{\smash{{\SetFigFont{25}{30.0}{\rmdefault}{\mddefault}{\updefault}{\color[rgb]{0,0,1}$a$}%
}}}}
\put(16191,-3509){\makebox(0,0)[rb]{\smash{{\SetFigFont{25}{30.0}{\rmdefault}{\mddefault}{\updefault}{\color[rgb]{1,0,0}$a_4'$}%
}}}}
\put(16596,-2189){\makebox(0,0)[lb]{\smash{{\SetFigFont{25}{30.0}{\rmdefault}{\mddefault}{\updefault}{\color[rgb]{1,0,0}$a_3'$}%
}}}}
\put(15276,-1994){\makebox(0,0)[rb]{\smash{{\SetFigFont{25}{30.0}{\rmdefault}{\mddefault}{\updefault}{\color[rgb]{0,0,1}$n_1$}%
}}}}
\put(2251,-12329){\makebox(0,0)[b]{\smash{{\SetFigFont{34}{40.8}{\rmdefault}{\mddefault}{\updefault}{\color[rgb]{0,0,0}(a)}%
}}}}
\put(7231,-12329){\makebox(0,0)[b]{\smash{{\SetFigFont{34}{40.8}{\rmdefault}{\mddefault}{\updefault}{\color[rgb]{0,0,0}(b)}%
}}}}
\put(12181,-12329){\makebox(0,0)[b]{\smash{{\SetFigFont{34}{40.8}{\rmdefault}{\mddefault}{\updefault}{\color[rgb]{0,0,0}(c)}%
}}}}
\put(17146,-12299){\makebox(0,0)[b]{\smash{{\SetFigFont{34}{40.8}{\rmdefault}{\mddefault}{\updefault}{\color[rgb]{0,0,0}(d)}%
}}}}
\end{picture}%

%% file: exits.pdf_t
\begin{picture}(0,0)%
\includegraphics{exits.pdf}%
\end{picture}%
%
%
\setlength{\unitlength}{4144sp}%
\begingroup\makeatletter\ifx\SetFigFont\undefined%
\gdef\SetFigFont#1#2#3#4#5{%
  \reset@font\fontsize{#1}{#2pt}%
  \fontfamily{#3}\fontseries{#4}\fontshape{#5}%
  \selectfont}%
\fi\endgroup%
\begin{picture}(14760,7566)(1276,-9885)
\put(5546,-2814){\makebox(0,0)[b]{\smash{{\SetFigFont{34}{40.8}{\rmdefault}{\mddefault}{\updefault}{\color[rgb]{0,0,0}${\cal F}_s$}%
}}}}
\put(5023,-6691){\makebox(0,0)[lb]{\smash{{\SetFigFont{25}{30.0}{\rmdefault}{\mddefault}{\updefault}{\color[rgb]{0,0,0}$r'$}%
}}}}
\put(4021,-5199){\makebox(0,0)[lb]{\smash{{\SetFigFont{25}{30.0}{\rmdefault}{\mddefault}{\updefault}{\color[rgb]{0,0,0}$r$}%
}}}}
\put(5506,-7314){\makebox(0,0)[lb]{\smash{{\SetFigFont{25}{30.0}{\rmdefault}{\mddefault}{\updefault}{\color[rgb]{0,0,0}$n_2$}%
}}}}
\put(7006,-8004){\makebox(0,0)[lb]{\smash{{\SetFigFont{25}{30.0}{\rmdefault}{\mddefault}{\updefault}{\color[rgb]{0,0,0}$n_3$}%
}}}}
\put(2806,-6159){\makebox(0,0)[rb]{\smash{{\SetFigFont{25}{30.0}{\rmdefault}{\mddefault}{\updefault}{\color[rgb]{1,0,0}$r_f'$}%
}}}}
\put(4831,-5559){\makebox(0,0)[rb]{\smash{{\SetFigFont{25}{30.0}{\rmdefault}{\mddefault}{\updefault}{\color[rgb]{1,0,0}$r_f'$}%
}}}}
\put(6856,-5574){\makebox(0,0)[lb]{\smash{{\SetFigFont{25}{30.0}{\rmdefault}{\mddefault}{\updefault}{\color[rgb]{1,0,0}$r_f'$}%
}}}}
\put(5401,-5184){\makebox(0,0)[lb]{\smash{{\SetFigFont{25}{30.0}{\rmdefault}{\mddefault}{\updefault}{\color[rgb]{0,0,0}$r$}%
}}}}
\put(6538,-5191){\makebox(0,0)[rb]{\smash{{\SetFigFont{25}{30.0}{\rmdefault}{\mddefault}{\updefault}{\color[rgb]{0,0,0}$r$}%
}}}}
\put(5591,-9744){\makebox(0,0)[b]{\smash{{\SetFigFont{29}{34.8}{\rmdefault}{\mddefault}{\updefault}{\color[rgb]{0,0,0}(a)}%
}}}}
\put(13511,-9699){\makebox(0,0)[b]{\smash{{\SetFigFont{29}{34.8}{\rmdefault}{\mddefault}{\updefault}{\color[rgb]{0,0,0}(b)}%
}}}}
\put(2656,-7329){\makebox(0,0)[lb]{\smash{{\SetFigFont{25}{30.0}{\rmdefault}{\mddefault}{\updefault}{\color[rgb]{0,0,0}$n_1$}%
}}}}
\put(6991,-4614){\makebox(0,0)[lb]{\smash{{\SetFigFont{25}{30.0}{\rmdefault}{\mddefault}{\updefault}{\color[rgb]{0,0,0}$n$}%
}}}}
\put(11821,-7016){\makebox(0,0)[rb]{\smash{{\SetFigFont{25}{30.0}{\rmdefault}{\mddefault}{\updefault}{\color[rgb]{1,0,0}$r_f'$}%
}}}}
\put(12136,-5231){\makebox(0,0)[lb]{\smash{{\SetFigFont{25}{30.0}{\rmdefault}{\mddefault}{\updefault}{\color[rgb]{0,0,0}$r$}%
}}}}
\put(13766,-2836){\makebox(0,0)[b]{\smash{{\SetFigFont{34}{40.8}{\rmdefault}{\mddefault}{\updefault}{\color[rgb]{0,0,0}${\cal F}_s$}%
}}}}
\put(11023,-6498){\makebox(0,0)[lb]{\smash{{\SetFigFont{25}{30.0}{\rmdefault}{\mddefault}{\updefault}{\color[rgb]{0,0,0}$r'$}%
}}}}
\put(12253,-7608){\makebox(0,0)[lb]{\smash{{\SetFigFont{25}{30.0}{\rmdefault}{\mddefault}{\updefault}{\color[rgb]{0,0,0}$r'$}%
}}}}
\put(14353,-7398){\makebox(0,0)[lb]{\smash{{\SetFigFont{25}{30.0}{\rmdefault}{\mddefault}{\updefault}{\color[rgb]{0,0,0}$r'$}%
}}}}
\put(9776,-6819){\makebox(0,0)[b]{\smash{{\SetFigFont{25}{30.0}{\rmdefault}{\mddefault}{\updefault}{\color[rgb]{0,0,0}$n_1$}%
}}}}
\put(12691,-8529){\makebox(0,0)[lb]{\smash{{\SetFigFont{25}{30.0}{\rmdefault}{\mddefault}{\updefault}{\color[rgb]{0,0,0}$n_2$}%
}}}}
\put(15766,-8289){\makebox(0,0)[lb]{\smash{{\SetFigFont{25}{30.0}{\rmdefault}{\mddefault}{\updefault}{\color[rgb]{0,0,0}$n_3$}%
}}}}
\put(3601,-6676){\makebox(0,0)[lb]{\smash{{\SetFigFont{25}{30.0}{\rmdefault}{\mddefault}{\updefault}{\color[rgb]{0,0,0}$r'$}%
}}}}
\put(15076,-4561){\makebox(0,0)[lb]{\smash{{\SetFigFont{25}{30.0}{\rmdefault}{\mddefault}{\updefault}{\color[rgb]{0,0,0}$n$}%
}}}}
\put(14491,-6451){\makebox(0,0)[rb]{\smash{{\SetFigFont{25}{30.0}{\rmdefault}{\mddefault}{\updefault}{\color[rgb]{0,0,0}$n'$}%
}}}}
\put(14671,-5236){\makebox(0,0)[rb]{\smash{{\SetFigFont{25}{30.0}{\rmdefault}{\mddefault}{\updefault}{\color[rgb]{0,0,0}$r$}%
}}}}
\put(13501,-5236){\makebox(0,0)[lb]{\smash{{\SetFigFont{25}{30.0}{\rmdefault}{\mddefault}{\updefault}{\color[rgb]{0,0,0}$r$}%
}}}}
\put(15406,-7031){\makebox(0,0)[lb]{\smash{{\SetFigFont{25}{30.0}{\rmdefault}{\mddefault}{\updefault}{\color[rgb]{1,0,0}$r_f'$}%
}}}}
\put(10711,-5471){\makebox(0,0)[rb]{\smash{{\SetFigFont{25}{30.0}{\rmdefault}{\mddefault}{\updefault}{\color[rgb]{1,0,0}$r_f'$}%
}}}}
\put(5611,-6121){\makebox(0,0)[lb]{\smash{{\SetFigFont{25}{30.0}{\rmdefault}{\mddefault}{\updefault}{\color[rgb]{0,0,0}$n'$}%
}}}}
\end{picture}%